\pdfoutput=1
\documentclass[journal]{IEEEtran}
\usepackage{mathrsfs}
\usepackage{pifont}
\usepackage{bbding}
\usepackage{multirow}
\usepackage{url}
\usepackage{amssymb}
\usepackage{stfloats}
\usepackage{amssymb}
\usepackage{makecell}
\usepackage{booktabs}
\usepackage{graphicx} 
\usepackage{epstopdf}
\usepackage{threeparttable}
\usepackage{graphicx}
\usepackage{amsthm}
\usepackage{balance}
\usepackage{amssymb, amsfonts, epsfig, latexsym, times}
\usepackage{bm}
\usepackage{color}
\usepackage{cite}
\usepackage{algorithm, algorithmic}
\usepackage{graphicx}
\usepackage{caption}
\usepackage[justification=raggedright,singlelinecheck=false,font=footnotesize]{caption}
\usepackage[subrefformat=parens,labelformat=simple,justification=centerlast]{subcaption}
\usepackage{float}
\usepackage{subcaption}
\usepackage[tbtags]{amsmath}

\IEEEoverridecommandlockouts
\def\BibTeX{{\rm B\kern-.05em{\sc i\kern-.025em b}\kern-.08em
    T\kern-.1667em\lower.7ex\hbox{E}\kern-.125emX}}
\newtheorem{lemma}{Lemma}
\newtheorem{remark}{Remark}

\usepackage{empheq}
\begin{document}
\captionsetup[figure]{labelsep=period,singlelinecheck=off,font=small}
\captionsetup[table]{
	labelsep = newline,
	singlelinecheck=true,
	font=small,
}
\title{Sensing Security in Near-Field ISAC: Exploiting Scatterers for Eavesdropper Deception}
\author{Jiangong Chen,
	Xia Lei,
	Kaitao Meng, \IEEEmembership{Member,~IEEE},
	Kawon Han, \IEEEmembership{Member,~IEEE},
	Yuchen Zhang, \IEEEmembership{Member,~IEEE},
	Christos Masouros, \IEEEmembership{Fellow,~IEEE}, and
	Athina P. Petropulu, \IEEEmembership{Fellow,~IEEE}
	\vspace{-10mm}
	
	\thanks{
		
		J. Chen and X. Lei are with the National Key Laboratory of Wireless Communications, University of Electronic Science and Technology of China, Chengdu, China (e-mail: jg\_chen@std.uestc.edu.cn, leixia@uestc.edu.cn).
		
		K. Meng is with the Department of Electrical and Electronic Engineering, University of Manchester, Manchester, UK (email: kaitao.meng@manchester.ac.uk).
		
		Y. Zhang is with the Electrical and Computer Engineering Program, Computer, Electrical and Mathematical Sciences and Engineering (CEMSE), King Abdullah University of Science and Technology (KAUST), Thuwal 23955-6900, Kingdom of Saudi Arabia (e-mail: yuchen.zhang@kaust.edu.sa).
		
		K. Han and C. Masouros are with the Department of Electronic and Electrical Engineering, University College London, London, UK (e-mail: \{c.masouros, kawon.han\}@ucl.ac.uk).
		
		A. Petropulu is with the Department of Electrical and Computer Engineering, Rutgers University, Piscataway, NJ 08901 USA (e-mail: athinap@soe.rutgers.edu).
		
	}
	
}

\maketitle
\begin{abstract}

In this paper, we explore sensing security in near-field (NF) integrated sensing and communication (ISAC) scenarios, by exploiting known scatterer in the sensing scene. We propose a location deception (LD) scheme where scatterers are deliberately illuminated with probing power that is higher than that directed towards targets of interest, with the goal of deceiving potential eavesdroppers (Eves) with sensing capability into misidentifying scatterers as targets.
While the known scatterers can be removed at the legitimate sensing receiver, our LD approach causes Eves to misdetect targets. Notably, this deception is achieved without requiring any prior information about the Eves’ characteristics or locations. To strike a flexible three-way tradeoff among communication, sensing and sensing-security performance, sum rate and power allocated to scatterers are weighted and maximized under certain legitimate radar signal-to-interference-plus-noise ratio (SINR). We employ the fractional programming (FP) framework and semidefinite relaxation (SDR) to solve this problem. To evaluate the security of the proposed LD scheme, the Cramér-Rao Bound (CRB) and mean squared error (MSE) metrics are employed.  Additionally, we introduce the Kullback-Leibler Divergence (KLD) gap between targets and scatterers at Eve to quantify the impact of the proposed LD framework on Eve's sensing performance from an information-theoretical perspective. Simulation results demonstrate that the proposed LD scheme can flexibly adjust the beamforming strategy according to performance requirements, thereby achieving the desired three-way tradeoff. Particularly in terms of sensing security, the proposed scheme significantly enhances the clutter signal strength at Eve’s side, leading to confusion or even missed detection of the actual target.

\end{abstract}

\IEEEpeerreviewmaketitle
\section{Introduction}
Integrated sensing and communication (ISAC) technology has been recognized as one of the six core scenarios of international mobile telecommunications (IMT)-2030 \cite{IMT2030}, garnering widespread attention from both academia and industry. By deeply integrating sensing and communication functions, ISAC enables the shared use of spectrum, hardware, and computational resources \cite{ISCABack}. This integration not only enhances overall system efficiency but also provides novel technological support for future intelligent applications, paving the way for revolutionary advancements in smart cities, intelligent transportation, industrial Internet, etc.  {\color{black}Although ISAC offers tremendous opportunities, addressing the accompanying security challenges before realizing these potentials is crucial to avoid data and geolocation privacy leakage.} Specifically, security concerns primarily focus on the confidentiality of communication signals in conventional communication systems \cite{PLS,PLSISACIRS}, whereas ISAC systems introduce entirely new vulnerabilities relating to geolocation privacy. 
For instance, {\color{black}sensing-capable eavesdroppers (Eves) may track the sensitive location information of transceivers and targets.} More concerningly, some advanced Eves might even exploit sensing results to recover the legitimate propagation channels, further enhancing their eavesdropping and malicious sensing capabilities, potentially leading to a vicious cycle. 
\vspace{-3mm}
\subsection{Related works}
\subsubsection{Communication security}
Early stage ISAC physical-layer security (PLS) techniques primarily focused on physical layer communication security and radar sensing performance \cite{ComSec1,ComSec2,ComSec3,ComSec4,ComSec5}, with Eves typically assumed to be the radar's sensing targets. Using radar signal-to-interference-plus-noise ratio (SINR) and secrecy rate as performance metrics, the authors of \cite{ComSec1} established three types of optimization problems, i.e., maximizing the security rate, maximizing radar SINR, and minimizing transmit power. However, \cite{ComSec1} considered perfect knowledge of Eve's channel state information (CSI), which is not a practical scenario. Considering the imperfect CSI of both the legitimate user and Eve (target) in \cite{ComSec2}, robust beamforming via fractional programming (FP), semidefinite relaxation (SDR), and S-procedure was conducted to minimize the signal-to-noise ratio (SNR) of the Eve under given multiple-input-multiple-output (MIMO) radar beampattern, communication SINR, and power budget constraints. To further improve the performance, directional modulation was utilized in \cite{ComSec3} to build constructive and destructive interference toward the legitimate user and Eve respectively. Furthermore, the details of Eve detection were given in \cite{ComSec4} through Capon and approximate maximum likelihood techniques. With the estimated directions of potential Eves, a weighted optimization problem was conducted to maximize the secrecy rate and simultaneously minimize the  Cramér-Rao bound (CRB) ratio between targets and Eves. In contrast to the above single-snapshot case, the authors in \cite{ComSec5} proposed a robust and secure ISAC scheme over a sequence of snapshots, where the transmit beamforming vector and artificial noise (AN) covariance matrix are jointly optimized under average achievable rate and  information leakage constraints in each time slot.

\subsubsection{Sensing security}
The aforementioned works have primarily focused on securing the communication aspect of transmission while overlooking the fact that Eves may also possess sensing capabilities. In such scenarios, the location privacy of legitimate transceivers and targets becomes compromised. Furthermore, Eve may leverage the eavesdropping location information to recover the legitimate statistical channel. By employing channel equalization and left-multiplying legitimate channel coefficients, Eve can receive signals as legitimate users, potentially breaking conventional PLS techniques such as AN \cite{ANE}. 

Related to our scope, but with the protection of transmitter location in mind, the authors of \cite{SenSec1,SenSec2} are the first to consider privacy protection of radar locations in scenarios of radar vs communications spectrum co-existence. In particular, it was shown that the precoder used by the radar to avoid interfering with the communication receiver can be exploited by Eve to infer the radar location. Subsequently, a novel precoder optimization problem was proposed in \cite{SenSec3} using gradient enforcement to achieve a trade off between the radar interference power and radar privacy. In order to achieve superior performance, the authors of \cite{SenSec4} considered neural-network-based optimization and  imaginary communication user. 

In addition to the aforementioned work on protecting radar privacy through precoder design, some other studies have attempted to introduce model mismatches at Eve to disrupt its sensing performance. For example, a CSI fuzzer was added to the WiFi systems to generate an artificial response which is known at the authorized receiver but prevents CSI-based sensing at an unauthorized receiver \cite{SenSec5}. Also focusing in the sensing privacy protection in WiFi systems, Aegis was introduced in \cite{SenSec6} as an innovative RF sensing shield that obfuscates human motion information for unauthorized receivers by distorting amplitude, delay, and Doppler shift via some hardware (a combination of amplifiers, a fan, and a directional antenna), while preserving legitimate sensing and communication performance. Similarly, \cite{SenSec7} introduced millimeter-filter (mmFilter), an application-oriented privacy filtering framework for  millimeter-wave (mmWave)  radar, which employs a signal reversion methodology to selectively perturb sensitive data while preserving permitted sensing functions. Tailored techniques and experiments verified that targeted low-level data modifications can effectively block unauthorized sensing without disrupting overall radar performance. {\color{black}Besides the sensing security}, to prevent unauthorized localization at the base station, a novel beamforming scheme was proposed in \cite{SenSec8,SenSecMa1,SenSecMa2} to nullify the line-of-sight (LoS) uplink signals and erase the angle-of-departure (AoD) information in non-LoS (NLOS) uplink signals, while optimizing the power allocation to maximize communication data rate. Numerical results demonstrate that this method achieves perfect location privacy with superior rate performance compared to zero-forcing beamforming, especially when the UE is equipped with a sufficient number of antennas. Inspired by these works, \cite{SenSec9} presented two pilot signal manipulation techniques, AN and artificial multipath, to protect user location privacy in time-difference-of-arrival (TDOA)-based localization systems by inducing model mismatch at unauthorized nodes via artificial multipath, significantly degrading their localization accuracy while minimally affecting legitimate localization.

Recently, sensing-capable Eves are considered in ISAC systems. Specifically, radar mutual information (RMI) was taken as the sensing metric in \cite{SenSec10}, where the RMI of the legitimate radar was maximized under the constraints of RMI for Eve and communication quality. Instead of RMI, detection probability was derived and optimized in \cite{SenSec11} for the both communication and sensing security in a cell-free ISAC system.  In addition to evaluating Eve's target detection performance, \cite{SenSec12} also employed the CRB to quantify its estimation performance and proposed optimization problems based on maximizing the missed detection probability and maximizing the CRB respectively to suppress malicious sensing.

\subsection{Motivations and Contributions}

{\color{black} Although PLS techniques have been widely studied in ISAC systems, most research focused exclusively on communication security, or on sensing security by protecting the location privacy of the transceivers. However, there is limited research on safeguarding the sensing security of targets' locations.}  In scenarios such as unmanned aerial vehicle (UAV) communications and vehicle-to-everything (V2X) networks, protecting privacy information about radar targets, including characteristic data and movement trajectories, is of significant importance for ensuring security and interests. Moreover, most research assumes that Eve’s CSI is either known or can be sensed by the legitimate transmitter and constructs optimizations based on Eve’s communication and sensing quality of service (QoS), {\color{black}which is not practical.} In reality, Eves are likely to remain silent and undetectable (e.g., staying far away from the legitimate radar, or using stealth materials \cite{Material}). Therefore, developing effective Eve-agnostic transmission methods under such scenarios is particularly important.

Furthermore, most existing studies consider a far-field model and  focus on a single dimension, primarily on the angle domain, for both communication and sensing. {\color{black}However, the deployment of extremely high frequency and extremely large antenna array (ELAA) inevitably leads to near-field (NF) effect in future ISAC systems. Different from the far-field model, NF channel, characterized by its unique spherical wave propagation, makes the channel dependent on both angle and distance \cite{NFSurvey1,NFSurvey2}, thereby enabling both wireless communication and sensing \cite{NFISAC,NFISAC2} across angle and distance dimensions. Some studies have already recognized the potential of leveraging the additional dimension of distance for position-based secure communications \cite{NFPLS} in NF scenarios. However, currently, no studies to date have considered the protection of both angle and distance privacy in NF ISAC systems.}


{\color{black}Against this background, we propose a location deception (LD) beamforming framework for NF ISAC systems to ensure both legitimate communication and sensing services while preventing Eve from successfully sensing the targets. }The main contributions are summarized as follows
 \begin{itemize}
 \item We introduce an NF-ISAC system for joint estimation of angle and distance of one or more targets in the presence of an unobservable Eve. To guarantee the sensing security, the LD concept is introduced to generate clutter interference. Different from some intelligent reflecting surface (IRS)/reconfigurable intelligent surface (RIS)-assisted ISAC systems \cite{ISACRIS1}, the signal reflection from the target is uncontrollable. Therefore, a transmitter side beamforming optimization framework is proposed to realize LD by optimizing the weighted sum of communication rate and power at the scatterers\footnote{Here, the known scatterers may be those that have already been detected by legitimate radars in the environment, or they could be artificially introduced.} under specific legitimate sensing SINR constraints, thereby posing sensing interference at Eve's side through clutter.

 \item  We observe that the rank of the receiver's covariance matrix is determined by the smallest value among the number of communication signal streams, the number of radar probing signal streams, and the number of reflectors (including target and scatterer sources). We utilize this observation to introduce a rank constraint on the radar signal covariance matrix that minimizes the rank of the covariance matrix of the signals received by Eve while ensuring the rank of the covariance matrix of the signal received by the legitimate radar is sufficient for target detection.
 This constraint significantly reduces the number of reflectors that Eve can estimate when the spatial smoothing (SS) technique is not employed, thereby limiting its ability to detect targets with lower energy levels. This non-convex constraint is equivalently transformed using the Courant-Fischer theorem and relaxed through the design of eigenvalues. Additionally, we apply fractional programming (FP) and semidefinite relaxation (SDR) to relax both the objective function and other constraints, thereby converting the problem into a solvable convex form.

 \item  To evaluate the security of the proposed LD scheme, we analyze the Cramér-Rao bounds (CRBs) of both legitimate radar and Eve to assess the lower bounds of their sensing performance. The CRB is further validated by comparing it with the mean squared error (MSE) performance of three classical estimators, i.e., the Capon estimator, the multiple signal classification (MUSIC) estimator, and the maximum likelihood estimator (MLE). {\color{black}Although Eve can achieve low CRB level, the MSE is significantly higher than the CRB due to the interference caused by the proposed LD scheme even when Eve employs SS.} Therefore, instead of the CRB, we introduce the Kullback-Leibler Divergence (KLD) gap between targets and scatterers at Eve to quantify the impact  of the proposed LD framework on Eve's sensing performance from an information-theoretical perspective. {\color{black}With the aid of the proposed LD scheme, Eve's KLD gap significantly decreases and exhibits negative growth with increasing SNR, indicating that Eve achieves better sensing performance at scatterers than at targets.}
 \end{itemize}

\subsection{Notations}
In this paper, we use lower-case letters, lower-case bold letters, and capital bold letters to denote scalars, vectors and matrices respectively. The operators for vectorization, transpose, conjugate, conjugate transpose, inverse, and Moore–Penrose inverse are denoted by  ${\rm vec}\left(\cdot\right)$, $\left(\cdot\right)^{\text{T}}$, $\left(\cdot\right)^{\text{H}}$, $\left(\cdot\right)^{\text{*}}$, $\left(\cdot\right)^{-1}$, and $\left(\cdot\right)^{\dag}$  respectively. $\text{Tr}\left(\mathbf{A}\right)$ and $\text{Rank}\left(\mathbf{A}\right)$ stand for the trace and rank of matrix $\mathbf{A}$. The $i$th element of vector $\mathbf{a}$ is $[\mathbf{a}]_i$, the $i$th column of matrix $\mathbf{A}$ is $\left[\mathbf{A}\right]_i$, and the $(i,j)$th element of matrix $\mathbf{A}$ is $[\mathbf{A}]_{i,j}$. $|a|$, $\|\mathbf{a} \|_2$, $\|\mathbf{A} \|_2$, and $\|\mathbf{A} \|_F$ respectively represent the modulus of scalar $a$, $\ell$-2 norm of vector $\mathbf{a}$, induced 2-norm of matrix $\mathbf{A}$, and Frobenius norm of matrix $\mathbf{A}$. {\color{black}$\mathbb{H}^+$ denotes the Hermitian semipositive definite matrix set. The union of sets $\mathcal{A}$ and $\mathcal{B}$ is represented as $\mathbf{A}\cup \mathcal{B}$. The set $\mathcal{A}$ with the element $a$ removed is represented as $\mathcal{A} \setminus a$.} The operations for extracting the imaginary part and real part of a complex variable are denoted by $\Im\left(\cdot\right)$ and $\Re\left(\cdot\right)$ respectively. $\otimes$ represents the Kronecker product. Moreover, $j = \sqrt{-1}$ is the imaginary unit. $\mathbf{I}_N$ denotes the $N$-dimensional unit array. Considering an optimization problem with respect to $x$, the optimal solution is denoted by $x^{\star}$.

\begin{figure}[ht]
	\centering
	\includegraphics[scale=0.9]{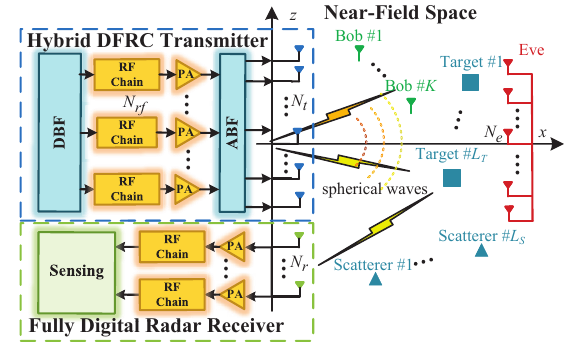}
	\caption{System model of the proposed LD scheme.}
	\label{SystemModel}
\end{figure}

\section{System Model}
\subsection{System and Channel Model}
As illustrated in Fig. \ref{SystemModel}, we consider a downlink multi-user (MU) multiple-input single-output (MISO) NF ISAC system comprising a dual function radar and communication (DFRC) transmitter (Alice), $K_c$ single-antenna communication receivers (Bobs), $L$ reflectors (including $L_T$ targets and $L_S$ scatterers), and an Eve with $N_e$-element uniform linear array (ULA). The hardware configuration of Alice comprises a {\color{black}DFRC transmitter and a colocated radar receiver} for communication and echo-based sensing {\color{black}in a full-duplex manner}. The transceiver is connected to ULAs with $N_t$ and $N_r$ antennas, respectively. To reduce costs, a hybrid architecture with $N_{rf}$ radio frequency (RF) chains is employed at the main transmitter, whereas the assisting receiver utilizes a fully digital architecture for sensing. Alice transmits ISAC signals to $K_c$ receivers, while simultaneously detecting the location information (angle $\theta$ and distance $r$) of the $L_T$ targets.\footnote{Since Alice is a monostatic radar setting, it has the same angle/distance of arrival and angle/distance of departure.} During this process, Eve remains passive and attempts to intercept the target location information embedded in the reflecting signals. In the following of this paper, we adopt several common assumptions in the ISAC and PLS fields as follows.

\begin{enumerate}
	\item Due to the high carrier frequency and large aperture of the ELAA, all receivers and reflectors are located in the NF region of both Alice and Eve, i.e., within the corresponding Rayleigh distance.
	\item The angle and distance information of each target and scatterer can be acquired by Alice in advance using estimation and tracking algorithms \cite{TarDec}.
	\item Eve remains passive, preventing Alice and Bob from obtaining its CSI and location. Additionally, since Eve cannot actively perform environment sensing, it lacks prior knowledge of any sensing-related information. {\color{black}Given that communication security in ISAC systems has been thoroughly researched \cite{ComSec1,ComSec2,ComSec3,ComSec4,ComSec5}, we assume that Eve is unable to decode the transmitted ISAC signals.}
\end{enumerate}

Without loss of generality, we position the center of Alice's transmitting and receiving array at the origin of the coordinate system $\left(0,0\right)$. The center of Eve's receiving array is located at Cartesian coordinate $\left(r_{E}\sin \theta_E,r_E\cos \theta_E\right)$, while the Cartesian coordinates of $K_c$ receivers of Bob are given by $\left(r_{B}^k\sin \theta_B^k,r_B^k\cos \theta_B^k\right),\;k \in \mathcal{K}_c\triangleq \left[1,2,\cdots,K_c\right]$, {\color{black}where $\theta$ and $r$ denote the angle and distance respect to the center of the transmitting array, i.e. the AoD and distance of distance of departure from Alice.} Similarly, the polar coordinates of the targets and scatterer are $\left(r_R^l,\theta_R^l\right), l \in \mathcal{T}\triangleq \left[1,2,\cdots,L_T\right]$ and $\left(r_R^l,\theta_R^l\right), l \in \mathcal{C}\triangleq \left[L_T+1,L_T+2,\cdots,L\right]$. 

Using the location information, the approximated near-field array response vector is modeled as
 \begin{equation}
	{\left[ {{{\bf{a}}_{  N}}\left( {\theta ,r} \right)} \right]_n} = {e^{ - j\frac{{2\pi }}{\lambda }\left( { - nd\cos \theta  + \frac{{{n^2}{d^2}{{\sin }^2}\theta }}{{2r}}} \right)}},
	\label{NFARV}
\end{equation}
where the index $n \in \left[-\left(N-1\right)/2,\cdots,\left(N-1\right)/2\right]$, and $d = \lambda/2$ represents the element spacing with $\lambda$ denoting the wavelength. According to \cite{NFSurvey1}, the  communication channel from Alice to the $k$th receiver of Bob is given by
\begin{equation}
{{\bf{h}}^k_{{{\rm{B}}}}} = \bar{\alpha} _B^k{{\bf{a}}_{{N_t}}}\left( {\theta _B^k,r_B^k} \right) +
\sum\limits_{l \in {\cal T} \cup {\cal C}} {\tilde \alpha _{B}^{l,k}} {{\bf{a}}_{{N_t}}}\left( {\theta _{R}^l,r_{R}^l} \right),
	\label{ComChan}
\end{equation}
 where $\bar{\alpha} _B^k$ represents the channel gain of the LoS path, modeled as $\bar{\alpha}_B^k = \sqrt{\rho_0}/r_{B}^k e^{-j 2\pi r_{B}^k/\lambda}$ with $\rho_0 = \lambda/4\pi$ being the reference pathloss. $\tilde{\alpha}_{B}^{l,k} = \sqrt{\rho_0}/\left(r_B^k + r_{R}^l\right) e^{-j 2\pi \left(r_B^k + r_{R}^l\right)/\lambda}$ denotes the reflection coefficients of the NLoS paths. The following channels have the same definitions of channel gains and reflection coefficients. The channel from Alice to Eve is expressed as
 \begin{equation}
{{\bf{H}}_{{E}}} = {\bf{H}}_{{E}}^{\rm LoS} +\sum\limits_{l \in {\cal T} \cup {\cal C}}{{\bf{H}}_{{E}}^l} ,\;{\kern 1pt} {\kern 1pt} 
 \end{equation}
 where 
 \begin{equation}
 {\left[ {{\bf{H}}_{{E}}^{\rm LoS}} \right]_{n,m}} = \bar{\alpha}_E {e^{ - j2\pi /\lambda \left\| {{\bf{r}}_A^n - {\bf{r}}_E^m} \right\|}},\;{\kern 1pt} 
 \end{equation}
 \begin{equation}
 {\bf{H}}_{{E}}^l = \tilde{\alpha} _{{E}}^l{{\bf{a}}_{{N_t}}}\left( {{\theta_{R}^l,r_{R}^l}} \right){\bf{a}}_{N_e}^{\rm{H}}\left( {{\tilde{\theta}_{R}^l,\tilde{r}_{R}^l}} \right).
 \end{equation}
 Here, the term $ \left\| {{\bf{r}}_A^n - {\bf{r}}_E^m} \right\|$ represents the distance between the $n$th transmit antenna of Alice and the $m$th receive antenna of Eve. $\left( {{\tilde{\theta}_{R}^l,\tilde{r}_{R}^l}} \right)$ denotes the angle of arrival (AoA) and distance of arrival at Eve for the $l$-th reflector, which is related to the departure parameters as follows
 \begin{equation}
 	\small
\left\{ \begin{array}{l}
	\tilde r = \sqrt {{{\left( {r\sin \theta  - {r_E}\sin {\theta _E}} \right)}^2} + {{\left( {r\cos \theta  - {r_E}\cos {\theta _E}} \right)}^2}}, \\
	\tilde \theta  = \arctan \left( {\frac{{r\cos \theta  - {r_E}\cos {\theta _E}}}{{\sqrt {{{\left( {r\sin \theta  - {r_E}\sin {\theta _E}} \right)}^2} + {{\left( {r\cos \theta  - {r_E}\cos {\theta _E}} \right)}^2}} }}} \right),
\end{array} \right.
 \end{equation}

 Finally, the round-trip channel for sensing and the self-interference (SI) channel \cite{SI} at Alice are given by
 \begin{equation}
 \small
\mathbf{H}_A = \sum\limits_{l \in {\cal T} \cup {\cal C}} {\bf{H}}_{{A}}^l = \sum\limits_{l \in {\cal T} \cup {\cal C}} \tilde{\alpha} _{\rm{A}}^l{{\bf{a}}_{{N_t}}}\left( \theta_R^l, r_R^l \right){\bf{a}}_{{N_r}}^{\rm{H}}\left(  \theta_R^l, r_R^l\right),
 \end{equation} 
 \begin{equation}
 {\left[ {{{\bf{H}}_{{\rm{SI}}}}} \right]_{n,m}} = \frac{\rho_{\rm SI} }{{{ \left\| {{\bf{r}}_A^n - {\bf{r}}_A^m} \right\|}}}{e^{ - j\frac{{2\pi }}{\lambda }{ \left\| {{\bf{r}}_A^n - {\bf{r}}_A^m} \right\|}}},
 \end{equation}
 where $\rho_{\rm SI}$ is a power normalization constant, and {\color{black}$ \left\| {{\bf{r}}_A^n - {\bf{r}}_A^m} \right\|$ represents the distance between the $n$th and $m$th transmit antenna of Alice.}

\subsection{Transmission Signal Model}
First, the transmitted ISAC waveform $\mathbf{X} \in \mathbb{C}^{N_t \times S}$ in $S$ frames/snapshots is formulated as
\begin{equation}
{\bf{X}}  = \mathbf{F}\mathbf{W}\mathbf{S} = {\bf{F}}\left( {{{\bf{W}}_{{c}}}{{\bf{S}}_{\rm{c}}} + {{\bf{W}}_{{r}}}{{\bf{S}}_{{r}}}} \right),
\label{TxSig}
\end{equation}
where ${\bf{F}} \in {^{{N_t} \times {N_{rf}}}}$ is the analog beamforming (ABF) matrix, $\mathbf{W} = \left[\mathbf{W}_c \; \mathbf{W}_r\right]$ represents the overall digital beamforming (DBF) matrix, comprising the communication DBF matrix $\mathbf{W}_c \in \mathbb{C}^{N_{rf}\times K_c}$ and radar sensing DBF matrix $\mathbf{W}_r \in \mathbb{C}^{N_{rf} \times K_r}$. It is noteworthy that the number of radar probing streams $K_r$ is important, and we will analyze and design it later in the optimization part. The communication signal $\mathbf{S}_c$ and radar sensing signal $\mathbf{S}_r$ are assumed to be independent complex Gaussian sources, i.e. $\mathbb{E}\left\{\mathbf{S}\mathbf{S}^{\rm H}\right\} = \mathbf{I}_{K_c+K_r}$, where $\mathbf{S} = \left[\mathbf{S}_c^{\rm T} \; \mathbf{S}_r^{\rm T}\right]^{\rm T}$. Thus, the sample transmit covariance matrix is given by
\begin{equation}
\begin{aligned}
\mathbf{R}_X = \frac{1}{S} \mathbf{X}\mathbf{X}^{\rm H} &= \mathbf{R}_c + \mathbf{R}_r \\
&\approx \mathbf{F}\mathbf{W}_c \mathbf{W}_c^{\rm H} \mathbf{F}^{\rm H} + \mathbf{F}\mathbf{W}_r \mathbf{W}_r^{\rm H} \mathbf{F}^{\rm H} 
\end{aligned}
\end{equation}
\subsubsection{Communication Model and Optimization Metrics}
Considering a single frame/snapshot in (\ref{TxSig}), the received signal at the $k$th Bob receiver is given by
\begin{equation}
\begin{aligned}
{y^k_{{{{B}}}}} &= {\bf{h}}_{{{{B}}}}^{k\;\rm{H}}{\bf{F}}{\left[ {{{\bf{W}}_{\rm{c}}}} \right]_{k}}{\left[ {{{\bf{s}}_{{c}}}} \right]_k} \\
&+ \sum\limits_{j \ne k} {{\bf{h}}_{{{\rm{B}}}}^{k\;\rm{H}}{\bf{F}}{{\left[ {{{\bf{W}}_{{c}}}} \right]}_{j}}{{\left[ {{{\bf{s}}_{\rm{c}}}} \right]}_j}}  + {\bf{h}}_{{{{B}}}}^{k\;\rm{H}}{\bf{F}}{{\bf{W}}_{{r}}}{{\bf{s}}_{{r}}} + {n^k_{{{{B}}}}},
\end{aligned}
\end{equation}
where ${n^k_{{{{B}}}}} \sim \mathcal{CN} \left(0,\sigma_{\rm n}^2\right)$ is the additive white Gaussian noise (AWGN) at the $k$th Bob receiver. For simplicity, we denote $\left[ {{{\bf{W}}_{\rm{c}}}} \right]_{k}$ as $\mathbf{w}_c^k$ and $\mathbf{R}_c^k = \mathbf{w}_c^k \mathbf{w}_c^{k\, {\rm H}}$ in the remainder of this paper. Accordingly, the $k$th Bob receiver's SINR is given by
\begin{equation}
{\gamma^k_{{{{B}}}}}  = \frac{{{{\left| {{\bf{h}}_{{{{B}}}}^{k\;\rm{H}}{\bf{Fw}}_c^{{k}}} \right|}^2}}}{{\sum\limits_{j \ne k} {{{\left| {{\bf{h}}_{{{{B}}}}^{k\;\rm{H}}{\bf{Fw}}_c^{{j}}} \right|}^2} + {{\left\| {{\bf{h}}_{{{\rm{B}}}}^{k\;\rm{H}}{\bf{F}}{{\bf{W}}_{{r}}}} \right\|}^2} + \sigma _{\rm{n}}^2} }}{\kern 1pt}.
\end{equation}
Then, the achievable rate of Bob is selected as the communication optimization metric, formulated as
\begin{equation}
{R_{{B}}} = \sum\limits_{k = 1}^K R_B^k =\sum\limits_{k = 1}^K {{{\log }_2}\left( {1 + {\gamma^k_{{{{B}}}}}} \right)} .
\end{equation}
\subsubsection{Radar Sensing Model and Metrics}
First, the received signals of the legitimate radar Alice and Eve are respectively given by
\begin{equation}
{{\bf{Y}}_{{A}}} = \overbrace {\sum\limits_{{l} \in {{\cal T}}}^{} {{\bf{H}}_{{A}}^{{l} \;{\rm{H}}}{\bf{X}}} }^{{\rm{Targets}}} + \overbrace {\sum\limits_{{l} \in {{\cal C}}}^{} {{\bf{H}}_{{A}}^{{l}\;{\rm{H}}}{\bf{X}}} }^{{\rm{Clutters}}} + \overbrace {{\bf{H}}_{{\rm{SI}}}^{\rm{H}}{\bf{X}}}^{{\rm{SI}}} + \overbrace {{{\bf{N}}_{{A}}}}^{{\rm{AWGN}}},
\end{equation}
\begin{equation}
{{\bf{Y}}_{{E}}} = \overbrace {\sum\limits_{{l} \in {\cal T}}^{} {{\bf{H}}_{{E}}^{{l}\;{\rm{H}}}{\bf{X}}} }^{{\rm{Targets}}} + \overbrace {\sum\limits_{{l} \in {{\cal C}}}^{} {{\bf{H}}_{{E}}^{{l} \;{\rm{H}}}{\bf{X}}} }^{{\rm{Clutters}}} + \overbrace {{\bf{H}}_{{E}}^{{\rm{LoS \; H}}}{\bf{X}}}^{{\rm{LoS}}} + \overbrace {{{\bf{N}}_{{E}}}}^{{\rm{AWGN}}},
\end{equation}
where $\mathbf{N}_A$ and $\mathbf{N}_E$ are AWGN matrices, with each element obeying $\mathcal{CN}\left(0,\sigma_{\rm n}^2\right)$. With the transmit signal, SI channel and scatterer channel, the interference components can be removed at Alice\footnote{\color{black}Based on the assumption that the scatterer and SI channels are time-invariant during coherent time, they can be readily subtracted from the received signal $\mathbf{Y}_A$ directly \cite{InterCancel}.}
, thus obtaining receive signals only involved with targets as 
\begin{equation}
{{\bf{Y}}_{{A}}}^{'} =  {\sum\limits_{{l} \in {{\cal T}}}^{} {{\bf{H}}_{{A}}^{{l} \;{\rm{H}}}{\bf{X}}} } +  {{{\bf{N}}_{{A}}}},
\label{AliceTM}
\end{equation}
Since the primary objective of this paper is to conduct LD at Eve using scatterers, the radar SINR at Alice and the power projected toward scatterers are two metrics suitable for optimization, where the radar SINR of the $l_t$th target is defined as
\begin{equation}
\small
	\eta_{{A}}^{l_t }= \frac{{{{\left| {\alpha _{{A}}^{l_t}} \right|}^2}\left\| {{\bf{a}}_{{N_t}}^{\rm{H}}\left( \theta_{R}^{l_t},r_{R}^{l_t} \right){\bf{FW}}} \right\|_2^2}}{{\sum\limits_{l \in \left( {{{\cal T}} \setminus l_t} \right)} {{{\left| {\alpha _{{A}}^l} \right|}^2}} \left\| {{\bf{a}}_{{N_t}}^{\rm{H}}\left( \theta_{R}^l, r_{R}^l \right){\bf{FW}}} \right\|_2^2 + \sigma _{\rm{n}}^2}}.
\end{equation}
Additionally, the power projected toward the $l_c$-th scatterer is expressed as
\begin{equation}
P_{C}^{l_c} = \left( \sqrt{\rho_0}/r_R^{l_c} \right)^2  
\left\| {{\bf{a}}_{{N_t}}^{\rm{H}}\left( \theta_{R}^{l_c},r_{R}^{l_c} \right){\bf{FW}}} \right\|_2^2.
\end{equation}
As previously discussed, Alice possesses prior knowledge of the transmit signals and location information of the scatterers, thus supporting advanced estimation methods and clutter interference cancellation. In contrast, Eve can only perform parameter estimation based on the power of echo signals. As such, it is meaningful to maximize the power projected toward scatterers under radar SINR constraints of Alice.



\subsection{The Number of Radar Sensing Signal Streams}
As demonstrated in \cite{DoF1,DoF2}, the radar signal is essential for ISAC systems when the number of the communication signal streams is less than that of sensing targets. Undoubtedly, optimal sensing performance is achieved when a radar signal with $N_{rf}$ streams is used, as demonstrated by the signal configuration in \cite{DoF1,DoF2}. However, it is observed that the number of radar sensing signal streams is related to the rank of Eve's received covariance matrix, which is regarded as the number of incoming sources in DoA estimation field \cite{MDLAIC}. Thus, instead of using radar signal with $N_{rf}$ streams, we select a radar sensing signal with a dimension of $K_r$, {\color{black}satisfying ${\rm rank}(\mathbf{R}^{'}_X) < L$}, to deceive Eve into misjudging the number of reflectors. Specifically, we first assume a receive beamforming at Eve as
\begin{equation}
{{\bf{P}}_{{E}}} = {\rm{Null}}\left( {{\bf{H}}_{{E}}^{{\rm{LoS H}}}}\; \right),
\end{equation}
which eliminates the uninterested LoS interference at Eve. {\color{black}Here, $ {\rm{Null}}\left( {{\bf{H}}_{{E}}^{{\rm{LoS H}}}}\; \right)$ denotes the null-space of $ {{\bf{H}}_{{E}}^{{\rm{LoS H}}}}$, which can be obtained using singular value decomposition (SVD) easily.} Next, the rank of the received covariance matrix at Alice is given by
\begin{equation}
{{\bf{R}}_{{A}}} = \mathbb{E} \left\{ {{{\bf{Y}}^{'}_{\rm{A}}}{\bf{Y}}_{{A}}^{'\;\rm{H}}} \right\} = {\bf{H}}_{{A}}^{{'}\;{\kern 1pt} {\rm{H}}}{\bf{R}}_{{X}}^{{'}}{\bf{H}}_{{A}}^{{'}} + \sigma _{\rm{n}}^2 {\bf I},
\end{equation}
where $\mathbf{R}_X^{'} ={{\bf{1}}_{{L_T} }} \otimes {\bf R}_X$, and ${\bf{H}}_{\rm{A}}^{{'}}$ is expressed as
\begin{equation}
	{\rm{ }}{\bf{H}}_{{A}}^{'\;{\kern 1pt} {\rm{H}}} = \left[ {\begin{array}{*{20}{c}}
			{{\bf{H}}_{{A}}^{1\;{\kern 1pt} {\rm{H}}}}& \cdots &{{\bf{H}}_{{A}}^{L_T\;{\kern 1pt} {\rm{H}}}}
	\end{array}} \right]
\end{equation} 
Under the noiseless scenario, we can prove that
\begin{equation}
\begin{aligned}
{\rm{rank}}\left( \mathbf{R}_A \right) &\leq \min \left\{ {{\rm{rank}}\left( {{\bf{H}}_{{A}}^{'}} \right),{\rm{rank}}\left( {{\bf{R}}_{{X}}^{'}} \right)} \right\} \\ 
&\leq \min \left\{ {{L_T},K_c+K_r} \right\}.
\end{aligned}
\end{equation}
Similarly, at Eve's side, the received covariance matrix is expressed as
\begin{equation}
{{\bf{R}}_{{E}}} = \mathbb{E} \left\{ {{{\bf{Y}}^{'}_{{E}}}{\bf{Y}}_{{E}}^{'\;\rm{H}}} \right\} = {\bf{H}}_{{E}}^{{'}\;{\kern 1pt} {\rm{H}}}{\bf{R}}_{{X}}^{{''}}{\bf{H}}_{{E}}^{{'}} + \sigma _{\rm{n}}^2{\bf{P}}_{{E}} {\bf{P}}_{{E}}^{\;{\kern 1pt} {\rm{H}}},
\end{equation}
where $\mathbf{R}_X^{''} = {{\bf{1}}_{{L+1} }} \otimes {\bf R}_X$, and ${\bf{H}}_{{E}}^{'\;{\kern 1pt} {\rm{H}}}$ is expressed as
\begin{equation}
\begin{array}{c}
	{\rm{ }}{\bf{H}}_{{E}}^{'\;{\kern 1pt} {\rm{H}}} = {{\bf{P}}_{{E}}}\left[ {\begin{array}{*{20}{c}}
			{{\bf{H}}_{{E}}^{1\;{\kern 1pt} {\rm{H}}}}& \cdots &{{\bf{H}}_{{E}}^{L\;{\kern 1pt} {\rm{H}}}}&{{\bf{H}}_{{{E}}}^{{\kern 1pt} {\rm{LoS \;H}}}}
	\end{array}} \right]\\
	\;\;\;\;\;= \left[ {\begin{array}{*{20}{c}}
			{{{\bf{P}}_{{E}}}{\bf{H}}_{{E}}^{1\;{\kern 1pt} {\rm{H}}}}& \cdots &{{{\bf{P}}_{{E}}}{\bf{H}}_{{E}}^{{L}\;{\kern 1pt} {\rm{H}}}}&{\bf{0}}
	\end{array}} \right].
\end{array}
\end{equation}
Thus, the rank of Eve's covariance matrix is bounded as
\begin{equation}
\begin{aligned}
	{\rm{rank}}\left( \mathbf{R}_E \right) &\leq \min \left\{ {{\rm{rank}}\left( {{\bf{H}}_{{E}}^{'}} \right),{\rm{rank}}\left( {{\bf{R}}_{{X}}^{''}} \right)} \right\} \\ 
	&\leq \min \left\{ {{L},K_c+K_r} \right\}.
\end{aligned}
\end{equation}
Here, the difference between Alice and Eve becomes evident. Compared to Alice, Eve requires transmit signals with higher rank to to estimate the number of all the reflectors using techniques such as Akaike information criterion (AIC), minimum description length (MDL) along with MUSIC, etc.  {\color{black}To inhibit the sensing-eavesdropping performance at Eve, the upper bound of the rank of Eve's covariance is constrained as $K_c + K_r  < L$, thus resulting in a mis-detection of the number of the reflectors.}


\begin{remark}
Rank deficiency indicates that the echo signals are highly coherent, which can be alleviated by the famous spatial smoothing technique \cite{SS}. Therefore, the rank constraint is effective against a weak Eve without SS but can be easily defused when Eve employs SS. 
Later in the simulation part, we will further demonstrate that SS is essential for improving estimation performance at Eve.
\end{remark}
%

\section{Optimization Problem Formulation}
In this section, we formulate a secure beamforming optimization problem for the proposed system based on the metrics introduced in the previous section.
\vspace{-3mm}

\subsection{Transmitting Covariance Matrix Optimization}
Since prior information about Eve's location or CSI is unavailable, the objective is to maximize Bob's communication rate and the minimum power allocated for LD, subject to legitimate radar SINR and power budget constraints, as follows
\begin{subequations}
\begin{align}
\mathcal{P}_1:\;&\mathop {\max }\limits_{{{\bf{R}}_X}} w_c {R_{{B}}} +w_{ss}\mathop {{\rm{min}}}\limits_{l \in {{\cal C}}} \left( {P_C^l} \right)\\
&{\rm{s}}{\rm{.t}}{\rm{.}}\; {{\bf{R}}_X} = {{\bf{R}}_{\rm{c}}} + {{\bf{R}}_r},\\
&\;{\kern 1pt} \;{\kern 1pt} \;{\kern 1pt} \;\eta _{\rm{A}}^l \geq {P_{{\rm{SINR}}}},\forall l \in {\cal T},\\
&\;{\kern 1pt} \;{\kern 1pt} \;{\kern 1pt} \;{\rm{Tr}}\left( {{{\bf{R}}_X}} \right) \le {P_T},\\
&\;{\kern 1pt} \;{\kern 1pt} \;{\kern 1pt} \;{\rm{Rank}}\left( {{{\bf{R}}_c^k}} \right) = 1,\forall k \in {\cal K}{_c},\\
&\;{\kern 1pt} \;{\kern 1pt} \;{\kern 1pt} \;{\rm{Rank}}\left( {{{\bf{R}}_{\rm{r}}}} \right) = \;{K_r},\\
&\;{\kern 1pt} \;{\kern 1pt} \;{\kern 1pt} \;{{\bf{R}}_c^k},{{\bf{R}}_{\rm{r}}} \in {\mathbb{H}^+ },\forall k \in {{\cal K}_c},
\end{align}
\label{opt1}%
\end{subequations}
{\color{black}where $w_c$ and $w_{ss}$ are the weighting factors for communication and LD power, satisfying $w_c + w_{ss} = 1$. ${P_{{\rm{SINR}}}} = w_s P_{\rm SINR}^{\rm max}$ represents the SINR QoS constraint, with $w_s \in [0,\;1]$ denoting the weight for sensing and $P_{\rm SINR}^{\rm max}$ representing the maximum radar SINR at Alice.} $P_T$ denotes the transmit power budget. By introducing a slack variable $\delta$, the max-min structure in the objective can be removed as
\begin{subequations}
	\begin{align}
		&\mathop {\max }\limits_{{{\bf{R}}_X}} w_c{R^{\rm FP}_{\rm{B}}} + w_{ss}\delta\\
		&{\rm{s}}{\rm{.t}}{\rm{.}}\; {P_C^l} \geq \delta,\; l\in \mathcal{C},\\
		&\;{\kern 1pt} \;{\kern 1pt} \;{\kern 1pt} \;(\ref{opt1}{\rm b}) - (\ref{opt1}{\rm g}).\nonumber
	\end{align}
\end{subequations}
Using the SDR technique, we temporarily remove (\ref{opt1}e) and (\ref{opt1}f). Furthermore, we apply the quadratic transform in \cite{FP} to address the sum-of-functions-of-ratio form of $R_B$ as
\begin{equation}
R^{\rm FP}_B = \sum\limits_{k = 1}^{K_c} {{\log_2}} \left( 1 + {2{\beta _k}\sqrt {{A_k}}  - \beta_k^2{B_k}} \right),
\end{equation}
where $\beta_k, k \in \mathcal{K}_c$ are auxiliary variables with optimal solution $\beta _k^{\star} = \sqrt{{{ A}_k^{\star}}} / {{ B}_k^{\star}}$. ${{A}}_k$ and ${{B}}_k$ are expressed as
\begin{equation}
{{A}}_k= {\rm Tr} \left(\mathbf{h}_{B_k}^{\rm H} \mathbf{R}_c^k \mathbf{h}_{B_k}  \right),
\end{equation}
\begin{equation}
{{B}}_k = \sum\limits_{j \ne k}^{K_c} {{\text{Tr}}\left( \mathbf{h}_{B_k}^{\rm H} \mathbf{R}_c^j \mathbf{h}_{B_k} \right)}  + 
{{\text{Tr}}\left( \mathbf{h}_{B_k}^{\rm H} \mathbf{R}_r \mathbf{h}_{B_k} \right)} 
+ \sigma _{\rm n} ^2.
\end{equation}
{\color{black}Subsequently, SDR is employed to achieve a convex semidefinite programming (SDP) as
\begin{subequations}
	\begin{align}
		\mathcal{P}_2:\;&\mathop {\max }\limits_{{{\bf{R}}_X}} w_c{R^{\rm FP}_{\rm{B}}} + w_{ss}\delta\\
		&{\rm{s}}{\rm{.t}}{\rm{.}}\; {P_C^l} \geq \delta,\; l\in \mathcal{C}\\
		&\;{\kern 1pt} \;{\kern 1pt} \;{\kern 1pt} \;(\ref{opt1}{\rm b}) - (\ref{opt1}{\rm d}),\; (\ref{opt1}{\rm g}), \nonumber
	\end{align}
	\label{opt2}%
\end{subequations}}which can be solved by optimizing $\mathbf{R}_X$ and $\beta_k$ in an alternating manner. 
\begin{lemma}
{\color{black}If problem $\mathcal{P}_2$ is feasible, we can always achieve optimal solutions $\mathbf{R}_c^{k\,\star}$ and $\mathbf{R}_r^{\star}$, with ${\rm rank}\left(\mathbf{R}_c^{k \, \star}\right) = 1,\, \forall k \in \mathcal{K}_c$ and $\mathbf{R}_r^{\star} = \mathbf{0}$.}
\label{lemma1}
\end{lemma}
\begin{proof}
	See \cite{Rank1Proof} and Appendix A of this paper.
\end{proof}
\begin{remark}
 SDP problems involving radar SINR always have a low-rank optimal solution under a limited number of constraints \cite{SDPRank}. A similar conclusion can be found in \cite{ZeroRadar}. However, if we focus on the direct estimation performance, e.g. minimizing the CRB \cite{DoF2}, the low-rank structure disappears and radar sensing covariance matrix is no longer zero.
\end{remark}
Next, we demonstrate how to obtain a radar signal transmitting covariance matrix with a fixed rank constraint in (\ref{opt1}f). Using the Courant-Fischer theorem, an extension of Rayleigh–Ritz theorem, we can prove that the rank constraint of a positive semidefinite matrix is equivalent to the following eigenvalue constraints
\begin{equation}
{\rm{Rank}}\left( {{{\bf{R}}_{\rm{r}}}} \right) = {K_r}\; \Leftrightarrow \left\{ \begin{array}{l}
	\sum\limits_{i = 1}^{{N_t} - {K_r}} {{\lambda _i} = 0} ,\\
	\sum\limits_{i = 1}^{{N_t} - {K_r} + 1} {{\lambda _i} \ge \kappa {P_S}/{K_r}\;},
\end{array} \right.
\label{rankconstraint}
\end{equation}
where $\kappa$ represents the power allocation between the communication signal and the radar sensing signal, and the eigenvalues of $\mathbf{R}_r$ are denoted as $\left\{ {{\lambda _1} \cdots {\lambda _{{N_t}}}} \right\}$ in an increasing order. Since summations of an arbitrary number of minimum eigenvalues here are concave, the second constraint on the right side of (\ref{rankconstraint}) is a convex one. To address the first nonconvex constraint, we adopt successive convex approximation (SCA) and rewrite it as follows
\begin{subequations}
\begin{align}
&\sum\limits_{i = 1}^{{N_t} - {K_r}} {{{\tilde \lambda }_i} + {\bf{\tilde v}}_i^{\rm{H}}\left( {{{\bf{R}}_{\rm{r}}} - {{{\bf{\tilde R}}}_{\rm{r}}}} \right){{{\bf{\tilde v}}}_i} \le 0} ,\\
&\sum\limits_{i = 1}^{{N_t} - {K_r} + 1} {{\lambda _i} \ge \kappa {P_S}/{N_{rf}}} ,
\end{align}
\label{rankcons}%
\end{subequations}
where $\tilde{\mathbf{R}}_r$ is the first-order Taylor expansion of $\mathbf{R}_r$ with $\left\{ {{{\tilde \lambda }_1} \cdots {{\tilde \lambda }_{{N_t}}}} \right\}$ and $\left\{ {{{{\bf{\tilde v}}}_1} \cdots {{{\bf{\tilde v}}}_{{N_t}}}} \right\}$ representing its increasing-order eigenvalues and corresponding eigenvectors. Therefore, the optimization problem ensuring a fixed-rank radar covariance matrix can be reformulated as
\begin{subequations}
	\begin{align}
		\mathcal{P}_3:\;&\mathop {\max }\limits_{{{\bf{R}}_X}} w_c{R^{\rm FP}_{\rm{B}}} + w_{ss}\delta\\
		&{\rm{s}}{\rm{.t}}{\rm{.}}\; (\ref{opt1}{\rm b}) - (\ref{opt1}{\rm d}),\; (\ref{opt1}{\rm g}), \;(\ref{opt2}{\rm b}),\; (\ref{rankcons}{\rm a}), \;(\ref{rankcons}{\rm b}).\nonumber
	\end{align}
	\label{opt3}%
\end{subequations}
This problem is a standard convex SDP problem and can be solved by the CVX toolbox. Note that the above procedure reduces to a power constraint if $K_r = 1$. Compared with problem $\mathcal{P}_2$,  problem $\mathcal{P}_3$ has two extra affine constraints in (\ref{rankcons}a) and (\ref{rankcons}b). Thus, based on Lemma \ref{lemma1}, rank-1 communication covariance matrices and rank-$K_r$ radar covariance matrix are guaranteed.

\subsection{Hybrid Beamforming Approximation}
With the optimal solution $\mathbf{R}_{c}^{k\;\star}$ and $\mathbf{R}_r^{\star}$ obtained from problem $\mathcal{P}_3$, their eigenvalue decomposition (EVD) given by $\mathbf{R}_{ck}^{\star} = \lambda_k \mathbf{u}_k\mathbf{u}_k^{\rm H}$ and $\mathbf{R}_{r}^{\star} =  \mathbf{U} \mathbf{\Lambda} \mathbf{U}^{\rm H}$, the fully digital beamformers can be achieved as
\begin{equation}
\begin{aligned}
&\mathbf{w}_{ck}^{\rm FD} = \sqrt{\lambda_k}\mathbf{u},\\
&\mathbf{W}_r^{\rm FD} = \sqrt{\mathbf{\Lambda}} \mathbf{U}.
\end{aligned}
\end{equation}
Then, the ABF and DBF matrices for communication are obtained by the following optimization
\begin{subequations}
\begin{align}
\mathcal{P}_4:\;&\mathop {\min }\limits_{\mathbf{F},\;\mathbf{W}_c} \|\mathbf{F}\mathbf{W}_c - \mathbf{W}_c^{\rm FD}\|_F^2,\\
&{\rm{s}}{\rm{.t}}{\rm{.}}\; \left|\left[\mathbf{F}\right]_{i,j}\right| = 1,\; \forall i,\;j,\\
&\;{\kern 1pt} \;{\kern 1pt} \;{\kern 1pt} \;  \; \|\mathbf{F}\mathbf{W}_c \|_F^2=\left(1-\kappa \right)P_S,
\end{align}
\label{MO}%
\end{subequations}
where $\mathbf{W}_c^{\rm FD} = \left[\mathbf{w}_{c1}^{\rm FD},\; \dots,\; \mathbf{w}_{cK_c}^{\rm FD}\right]$. This is a classic approximation in hybrid beamforming field \cite{HBF}, which can be solved by manifold optimization method given in \cite{MO}. With the optimal ABF matrix $\mathbf{F}^{\star}$, the least square solution for the DBF matrix of the radar sensing signal is given by
\begin{equation}
\mathbf{W}_r^{\star} = \mathbf{F}^{\dag} \mathbf{W}_r^{\rm FD}.
\label{LS}
\end{equation}
\subsection{Complexity Analysis}
Based on the above three subsections, the detailed solution process for problem $\mathcal{P}_1$ is summarized in Algorithm \ref{Alg1}. After relaxation, problem $\mathcal{P}_1$  is a standard SDP problem, which can be solved via the interior point method with a complexity order of ${\cal O}\left( {\max {{\left( {m,n} \right)}^4}{n^{1/2}}\log \left( {1/\varepsilon } \right)} \right)$ \cite{SDR}, where $n$ represents the dimension of the optimization variable, $m$ denotes the number of equality and inequality constraints, and $\varepsilon$ is the solution accuracy. For problem $\mathcal{P}^4$, the complexity primarily depends on the computation of the Euclidean gradient and least square, given by ${\cal O}\left(N_t^2 K_c N_{rf}\right)$. Thus the overall complexity of Algorithm \ref{Alg1} is of the order of {\small${\cal O}\left\{I_{\rm FP}I_{\rm SCA}  \max\left(K_c+1,L\right)^4 L^{1/2}\log \left( 1/\varepsilon  \right)  +  I_{\rm MO}N_t^2 K_c N_{rf}\right\}$}.

\begin{algorithm}[]
	\caption{Optimization Algorithm of problem $\mathcal{P}_1$.}\label{Alg1}
	\begin{algorithmic}[1]
		\REQUIRE ~~\\
		Environment parameters $L_T,\;L_C$, signal dimension $K_c,\;K_r$, FP and SCA iteration number $I_{\rm FP},\; I_{\rm SCA}$, initial covariance matrix $\mathbf{R}_c$, $\mathbf{R}_r$,
		\ENSURE ~~\\
		Optimal hybrid beamformer $\mathbf{F}^{\star}$, $\mathbf{W}_c^{\star}$, $\mathbf{W}_r^{\star}$;\\
		
		\FOR{$i=0$; $i<I_{\text{FP}}$; $i++$ }
		\STATE update ${\beta}_k$; \\
		\FOR{$j=0$; $j<I_{\text{SCA}}$; $j++$ }
		\STATE update $\tilde{\mathbf{R}}_r$; \\
		\STATE Calculate optimal covariance $\mathbf{R}_c^{\star}$ and $\mathbf{R}_r^{\star}$ according to problem (\ref{opt3}); \\
		\ENDFOR
		\ENDFOR
		\STATE Given the optimal solution $\mathbf{R}_c^{\star}$ and $\mathbf{R}_r^{\star}$, calculate the optimal hybrid beamformer $\mathbf{F}^{\star}$, $\mathbf{W}_c^{\star}$ and $\mathbf{W}_r^{\star}$ according to $\mathcal{P}_4$ and (\ref{LS}).
	\end{algorithmic}
\end{algorithm}

\section{Location Parameters Estimators}
{\color{black}In this section, we discuss three types of parameter estimators, Capon, MUSIC and MLE. Capon's method \cite{Capon}, also known as minimum variance distortionless response (MVDR), describes the received beampattern of the receiver, formulated as
\begin{equation}
	P_{\rm Cap} \left({\theta},{r}\right)= \frac{1}{\mathbf{a}^{\rm H}\left({\theta},{r}\right) \mathbf{R}^{-1}_{y} \mathbf{a}\left({\theta},{r}\right)},
\end{equation}
where $\mathbf{R}_y = \mathbf{Y}\mathbf{Y}^{\rm H}/S$ represents the sampled receive covariance matrix. In contrast, the MUSIC method does not rely on the power of the incoming waves but utilizes the orthogonality between the signal subspace and the noise subspace, generally achieving better performance in low SNR conditions. The two dimensional (2-D) MUSIC spectrum in NF is given by
\begin{equation}
	{P_{{\rm{MUSIC}}}}\left( { \theta , r} \right) = \frac{1}{{{\bf{a}}^{\rm{H}}{{\left( { \theta , r} \right)}}{{\bf{U}}_n}{\bf{U}}_n^{\rm{H}}{\bf{a}}\left( { \theta , r} \right)}},
\end{equation}
where the noise subspace $\mathbf{U}_n$ can be obtained by performing EVD on $\mathbf{R}_y$ as
\begin{equation}
	{{\bf{R}}_y} = \left[ {\begin{array}{*{20}{c}}
			{{{\bf{U}}_s}}&{{{\bf{U}}_n}}
	\end{array}} \right]
	\left[ {\begin{array}{*{20}{c}}
			{{{\bf{\Lambda }}_s}}&{}\\
			{}&{{{\bf{\Lambda }}_n}}
	\end{array}} \right]
	\left[ {\begin{array}{*{20}{c}}
			{{\bf{U}}_s^{\rm{H}}}\\
			{{\bf{U}}_n^{\rm{H}}}
	\end{array}} \right].
\end{equation}
If assuming receive beamforming (RBF) matrix $\mathbf{P}$, for example Eve's RBF in (20), the beamspace-MUSIC (BMUSIC) spectrum \cite{BMUSIC} is formulated as
\begin{equation}
	{P_{{\rm{BMUSIC}}}}\left( { \theta , r} \right) = \frac{{{{\left[ {{\bf{Pa}}\left( { \theta , r} \right)} \right]}^{\rm{H}}}\left[ {{\bf{Pa}}\left( { \theta , r} \right)} \right]}}{{{{\left[ {{\bf{Pa}}\left( { \theta , r} \right)} \right]}^{\rm{H}}}{{\bf{U}}^{'}_n}{\bf{U}}_n^{'\,\rm{H}}\left[ {{\bf{Pa}}\left( { \theta , r} \right)} \right]}},
\end{equation}
{\color{black}where the beamspace noise subspace $\mathbf{U}_n^{'}$ is obtained by the beamspace received covariance $\mathbf{R}_y^{'}=\mathbf{PY}\mathbf{Y}^{\rm H}\mathbf{P}^{\rm H}/S$.} The final estimation results can be obtained by searching for spectral peaks in the spatial spectrum generated by the aforementioned methods, based on the estimated number of reflectors.

It is noteworthy that both Capon and MUSIC do not require the knowledge of the transmit signal $\mathbf{X}$, thus can be employed at both Alice and Eve. Compared to Eve, Alice knows the transmit signal, allowing better estimation performance through MLE as follows
\begin{equation}
	\small
	\left( {{{{\bm{\hat \alpha }}}_A},{\bm{\hat \theta }},\hat {\bf r}} \right) = \arg \min_{\bm{{\alpha}}_A, \bm{\theta},\bf{r}} {\left\| {\mathbf{Y}}_A^{'} - \sum\limits_{l \in {\cal T}}\mathbf{H}_A^l \left( {{\bm{\alpha }_A},{\bm{\theta }},{\mathbf{r}}} \right) \mathbf{X}\right\|_F^2}.
	\label{MLE}
\end{equation}
}
However, the complexity of MLE is high, especially for the multiple targets scenario. To reduce complexity, the suboptimal Relax-MLE method proposed in \cite{RelaxMLE} is considered. Interested readers may refer to Section IV of \cite{RelaxMLE} for further details. {\color{black}It's noteworthy that the multi-target Relax-MLE requires knowledge of the transmitted signal, and therefore can only be employed by Alice.}
\vspace{-3mm}

\section{Performance Analysis}
In this section, we first introduce the KLD from an information-theoretic perspective to evaluate the sensing performance of both the legitimate radar Alice and the Eve. Specifically, we define a KLD gap between the targets and the scatterers for Eve to assess the effectiveness of the proposed LD scheme. Subsequently, we also define two MSE computation methods to evaluate the Eve's sensing performance under both actual and worst-case scenarios.



\subsection{Kullback–Leibler Divergence}
The KLD, also known as the relative entropy, measures the difference between a pair of random probability density functions (PDFs) $f_0\left(x\right)$ and $f_1\left(x\right)$, defined as
\begin{equation}
	D\left(f_0\parallel f_1\right)=\int_{-\infty}^{\infty}f_0\left(x\right)\log\left(\frac{f_0\left(x\right)}{f_1\left(x\right)}\right)dx.
\end{equation}
Existing works \cite{KLD1,KLD2} demonstrated that the KLD generalizes the concept of the mutual information and is also highly related to the detection probability. Let's first recall the received signals at Alice and Eve from the $l$th reflector as follows
\begin{equation}
\mathbf{Y}_A^{l} = \mathbf{H}_A^{l \; {\rm H}} \mathbf{X} + \mathbf{N}_A, \; l \in \mathcal{T}
\end{equation}
\begin{equation}
	\mathbf{Y}_E^{l} = \mathbf{H}_E^{l \; {\rm H}} \mathbf{X} + \mathbf{N}_E, l \; \in \mathcal{T} \cup \mathcal{C}
\end{equation}
Ignoring the subscripts, both Alice and Eve can formulate the hypothesis testing problem as follows
\begin{equation}
\left\{ \begin{gathered}
	{\mathcal{H}_0}:{\tilde{y}^l} = {{\mathbf{r}}^{l\;\text{H}}}{\text{vec}}\left( {\mathbf{N}} \right), \hfill \\
	{\mathcal{H}_1}:{y^l} = {{\mathbf{r}}^{l\;\text{H}}}{\text{vec}}\left( {{{\mathbf{H}}^{l\;{\kern 1pt} {\text{H}}}}{\mathbf{X}} + {\mathbf{N}}} \right), \hfill \\ 
\end{gathered}  \right.
\label{hypo}
\end{equation}
where $\mathbf{r}^{l\;{\rm H}}$ is the received beamforming vector for the $l$th reflector, given by ${\mathbf{r}}_A^l = {{\mathbf{1}}_S} \otimes {\mathbf{a}}_{{N_r}}^{}\left( {\theta _R^l,r_R^l} \right)$ at Alice's side and ${\mathbf{r}}_E^l = {{\mathbf{1}}_S} \otimes {\mathbf{a}}_{{N_e}}^{}\left( {\tilde \theta _R^l,\tilde r_R^l} \right)$ at Eve's side. It is straightforward to see that all the results in (\ref{hypo}) follow a complex Gaussian with zero mean but different variances as ${y}_A^l \sim {\cal CN}\left(0,\sigma_A^{l\;2}\right)$, ${\tilde{y}}_A^{l\;2} \sim {\cal CN}\left(0,\tilde{\sigma}_A^{l\;2}\right)$, ${y}_E^l \sim {\cal CN}\left(0,\sigma_E^{l\;2}\right)$, ${\tilde{y}}_E^l \sim {\cal CN}\left(0,\tilde{\sigma}_E^{l\;2}\right)$, where $\tilde{\sigma}_A^{l\;2} = S N_r\sigma _{\text{n}}^2$, $\sigma_A^2= S{\mathbf{r}}_A^{l\;\text{H}} {{\mathbf{H}}_A^{l\;{\kern 1pt} {\text{H}}}} {{\mathbf{R}}_X} {{\mathbf{H}}_A^l} {\mathbf{r}}_A^{l} + SN_r\sigma _{\text{n}}^{2}$, $\tilde{\sigma}_E^{l\;2} =  S N_e\sigma _{\text{n}}^2$, and $\sigma_E^{l\;2} = S{\mathbf{r}}_E^{l\;\text{H}}{{\mathbf{P}}_E} {{\mathbf{H}}_E^{l\;{\kern 1pt} {\text{H}}}} {{\mathbf{R}}_X} {{\mathbf{H}}_E^l} {\mathbf{P}}_E^{\text{H}}{\mathbf{r}}_E^{l} + S N_e\sigma _{\text{n}}^2$.
Therefore, the KLD between the two hypothesis testing observations at both Alice and Eve can be expressed as follows
\begin{equation}
{D^{l}_{A/E}} = \frac{1}{{\ln 2}}\left( {\frac{{\tilde \sigma _{A/E}^{l\;2}}}{{\sigma _{A/E}^{l\;2}}} - 1 + \ln \frac{{\sigma _{A/E}^{l\;2}}}{{\tilde \sigma _{A/E}^{l\;2}}}} \right).
\end{equation}
Considering all the targets at Alice, the averaged KLD of the targets at Alice is $D^{\rm ave}_A = \frac{1}{L_T}\sum_{l \in {\cal T}} D^l_A$. For Eve, the averaged KLD of targets and scatterers is $D^{\rm ave}_E = \frac{1}{L_T}\sum_{l \in {\cal T}} D^l_E$ and $\tilde{D}^{\rm ave}_E = \frac{1}{L_C}\sum_{l \in {\cal C}} D^l_E$, respectively. To quantify the negative impact of scatterers on target detection at Eve, the gap between targets and scatterers KLD at Eve is defined as
\begin{equation}
D_E^{\rm gap} =D^{\rm ave}_E  - \tilde{D}^{\rm ave}_E.
\end{equation}

\subsection{Root MSE and Root CRB}
In addition to the KLD, we use root MSE (RMSE) and root CRB (RCRB) to evaluate the sensing secrecy from the perspective of location estimation error. 

\subsubsection{RMSE}
{\color{black}First, the RMSE between the real location parameters $\{\bm{\theta},\mathbf{r}\}=\{\theta_R^1,\,\dots,\,\theta_R^{L_T},\,r_R^1,\,\dots,\,r_R^{L_T}\}$ and the corresponding estimation results $\{\hat{\bm{\theta}},\hat{\mathbf{r}}\}=\{\hat{\theta}_R^1,\,\dots,\,\hat{\theta}_R^{\hat{L}_T},\,\hat{r}_R^1,\,\dots,\,\hat{r}_R^{\hat{L}_T}\}$ can be formulated as}
\begin{equation}
{\rm{RMS}}{{\rm{E}}_{{\rm{ang}}}}{\rm{ = }}\sqrt { \frac{1}{S}\sum\limits_{s= 1}^S {\left\| {{\bm{\hat \theta }}\left( {\cal E}_1 \right) - {\bm{\theta }}}\left({\cal E}_2\right) \right\|_2^2} } ,
\end{equation}
\begin{equation}
{\rm{RMS}}{{\rm{E}}_{{\rm{dis}}}}{\rm{ = }}\sqrt {\frac{1}{S}\sum\limits_{s = 1}^S {\left\| {{\bf{\hat r}}\left( {\cal E}_1 \right) - {\bf{r}}}\left({\cal E}_2\right) \right\|_2^2} },
\end{equation}
where $\mathcal{E}_1$ and $\mathcal{E}_2$ are index sets used to select $L_{\rm MSE} = \min (L_T,\,\hat{L}_T)$ results from the estimated set $\{\hat{\bm{\theta}},\hat{\mathbf{r}}\}$ and the ground truth one $\{\bm{\theta},\mathbf{r}\}$ for RMSE calculation. For Alice, the estimated number of reflectors equals the real one, thus $\mathcal{E}^A_1  = \mathcal{E}^A_2 = \left\{1,\;2,\;\cdots,\;L_T\right\}$. For Eve, the estimated number of reflectors may not match that of the targets (we will show that in the simulation results part later).  {\color{black}Then, we define $\mathcal{E}_1^{E,{\rm min}}$ and $\mathcal{E}_2^{E,{\rm min}}$  as the index sets at Eve's side that minimizes the RMSE based on the accurate ground truth of the targets' locations,} and let $\mathcal{E}_1^{E,{\rm rand}}$ and $\mathcal{E}_2^{E,{\rm rand}}$ represent random selection index sets. We use the former two index sets to represent the Eve's sensing performance in the {\color{black}hypothetical} worst-case scenario, while the latter two random sets denote Eve's performance when it fails to distinguish between the targets and the scatterers in the actual scenario.

\subsubsection{RCRB}
The CRB, as the lower bound of the MSE, is analyzed here to evaluate the lower bound of sensing performance. For Alice, the parameters to be estimated are arranged as $\bm{\zeta}_A = \left\{\bm{\theta}_R,\; {\bf r}_R,\; \Re\left(\bm{{\alpha}}_A\right),\;\Im\left(\bm{{\tilde{\alpha}}}_A\right)\right\} $. Since the observation ${\rm vec}\left(\mathbf{Y}_A^{'}\right)$ is a Gaussian vector with mean of ${\bm \mu}\left({
\bm \zeta}_A\right)= {\rm vec}\left[{\sum\limits_{{l} \in {{\cal T}}}^{} {{\bf{H}}_{{A}}^{{l} \;{\rm{H}}}{\bf{X}}} }\right]$ and variance $\mathbf{R}\left({\bm \zeta}_A\right) = \sigma_{\rm{n}}^2\mathbf{I}$, the Fisher information matrix (FIM) \cite{CRB} can be expressed as
\begin{equation}
\begin{aligned}
&{\left[ {{\rm{FIM}}({\bm{\zeta }_A})} \right]_{p,q}} = \frac{2}{\sigma _{\rm{n}}^2} \Re \left[\frac{{\partial {\bm{\mu }}{{\left( {\bm{\zeta }_A} \right)}^{\rm{H}}}}}{{\partial {{\left[ {\bm{\zeta }_A} \right]}_p}}}\frac{{\partial {\bm{\mu }}\left( {\bm{\zeta }_A} \right)}}{{\partial {{\left[ {\bm{\zeta }_A} \right]}_q}}}\right]\\
&= \frac{{2S}}{{\sigma _{\rm{n}}^2}}\Re \left\{ {{\rm{Tr}}\left[ {\frac{{\partial \sum\limits_{l \in {\cal T}} {\bf{H}}_A^{l \;\rm{H}}}}{{\partial {{\left[ {{{\bm \zeta} _A}} \right]}_p}}}{{\bf{R}}_X}\frac{{\partial \sum\limits_{l \in {\cal T}} {\bf{H}}_A^{l}}}{{\partial {{\left[ {{{\bm \zeta} _A}} \right]}_q}}}} \right]} \right\},\;
\end{aligned}
\label{FIMA}
\end{equation}
The results of (\ref{FIMA}) for different parameters in $\bm{\zeta}_A$ are provided as follows
\begin{equation}
	\begin{aligned}
		&{\left[ {\frac{\partial \mathbf{H}_A^{l\;{\rm H}}}{{\partial {\theta _R^l}}}} \right]_{p,q}} = - {\left[\mathbf{H}_A^{l\; {\rm H}} \right]_{p,q}} \times  \\
		&\frac{{  j2\pi }}{\lambda }\left[ {\left( {\tilde{p} - \tilde{q}} \right)d\sin {\theta _R^l} + \frac{{\left( {{\tilde{p}^2} - {\tilde{q}^2}} \right){d^2}\sin {\theta _R^l}\cos {\theta _R^l}}}{{{r_R^l}}}} \right],
	\end{aligned}
	\label{Jac1}
\end{equation}
\begin{equation}
	{\left[ {\frac{\partial \mathbf{H}_A^{l \; {\rm H}}}{{\partial {r_R^l}}}} \right]_{p,q}} = {\left[\mathbf{H}_A^{{\rm H}} \right]_{p,q}} \times \frac{{j\pi \left( {{\tilde{p}^2} - {\tilde{q}^2}} \right){d^2}{{\sin }^2}{\theta _R^l}}}{{\lambda r_R^{l\;2}}},
	\label{Jac2}
\end{equation}
\begin{equation}
	\frac{{\partial {\bf{H}}_A^{l\;{\rm{H}}}}}{{\partial \Re \left( {\tilde \alpha _{\rm{A}}^{l\:}} \right)}} = {\bf{H}}_A^{l\;{\kern 1pt} {\rm{H}}}/\tilde \alpha _{\rm{A}}^{l\;{\kern 1pt} *\:},
	\label{Jac3}
\end{equation}
\begin{equation}
\frac{{\partial {\bf{H}}_A^{l\; {\rm{H}}}}}{{\partial \Im \left( {\tilde \alpha _{\rm{A}}^{l\:}} \right)}} =  - j{\bf{H}}_A^{l\;{\kern 1pt} {\rm{H}}}/\tilde \alpha _{\rm{A}}^{l\;{\kern 1pt} *\:},
	\label{Jac4}
\end{equation}
where $\tilde{p} = -\left(N_r - 1\right)/2 + p -1$ and $\tilde{q} = -\left(N_t - 1\right)/2 + q -1$. Thus, the CRB matrix is obtained as ${\rm CRB}\left({\bm \zeta}_A\right) = {\rm FIM}^{-1}\left({\bm \zeta}_A\right)$, and the RCRBs of different parameters are obtained by calculating the traces of different submatrices in ${\rm CRB}\left({\bm \zeta}_A\right) $.

For Eve, all the parameters to be estimated are arranged as $\bm{\zeta}_E = \left\{\bm{\tilde{\theta}}_R,\; {\bf \tilde{r}}_R,\; \Re\left(\bm{\tilde{\alpha}}_E\right),\;\Im\left(\bm{\tilde{\alpha}}_E\right)\right\} $. Unlike Alice's CRB calculation, which considers only the targets, Eve's CRB calculation needs to include both the targets and scatterers. The observation ${\rm vec}\left(\mathbf{Y}_E\right)$ is also a Gaussian vector but with zero mean and variance $\mathbf{R}^{'}\left({\bm \zeta}_E\right) = \mathbf{I}_S \otimes \mathbf{R}\left({\bm \zeta}_E\right)$, where $ \mathbf{R}\left({\bm \zeta}_E\right) = {\bf{H}}_E^{' \; \rm{H}}{{\bf{R}}_X}{{\bf{H}}_E^{'}}+ \sigma_{\rm{n}}^2\mathbf{P}_E\mathbf{P}_E^{\rm H}$ with $\mathbf{H}_E^{'} =\mathbf{H}_E \mathbf{P}_E^{\rm H}$. Thus, the FIM at Eve's side can be formulated as
\begin{equation}
\begin{aligned}
&{\left[ {{\rm{FIM}}({\bm{\zeta }_E})} \right]_{p,q}} = \\
&S {\rm Tr} \left[ {\bf R}^{-1} {{\left( {\bm{\zeta }_E} \right)}} \frac{{\partial {\bf{R }}{{\left( {\bm{\zeta }_E} \right)}}}}{{\partial {{\left[ {\bm{\zeta }_E} \right]}_p}}} {\bf R}^{-1} {{\left( {\bm{\zeta }_E} \right)}} \frac{{\partial {\bf{R }}\left( {\bm{\zeta }_E} \right)}}{{\partial {{\left[ {\bm{\zeta }_E} \right]}_q}}} \right],
\end{aligned}
\label{FIME}
\end{equation}
where 
\begin{equation}
\frac{{\partial {\bf{R}}\left( {{{\bm{\zeta }}_E}} \right)}}{{\partial {{\left[ {{{\bm{\zeta }}_E}} \right]}_p}}} =\frac{{\partial {\bf{H}}_E^{'\;\rm{H}}}}{{\partial {{\left[ {{{\bm{\zeta }}_E}} \right]}_p}}}{{\bf{R}}_X}{{\bf{H}}_E^{'}}+ {\bf{H}}_E^{' \; \rm{H}}{{\bf{R}}_X}\frac{{\partial {{\bf{H}}_E^{'}}}}{{\partial {{\left[ {{{\bm{\zeta }}_E}} \right]}_p}}} .
\end{equation}
Similarly, for different parameters in $\bm{\zeta}_E$, we have
\begin{equation}
\begin{aligned}
&{\left[ {\frac{{\partial {\mathbf{H}}_E^{'\;{\text{H}}}}}{{\partial \tilde \theta _R^l}}} \right]_{p,q}} = {\left[\mathbf{H}_E^{l} \mathbf{P}_E^{\rm H} \right]^{\rm H}_{p,q}} \times \\
&- j\frac{{2\pi }}{\lambda }\left( {\tilde{p}d\sin \tilde \theta _R^l + \frac{{{\tilde{p}^2}{d^2}}}{{\tilde r_R^l}}\sin \tilde \theta _R^l\cos \tilde \theta _R^l} \right),
\end{aligned}
\label{Jac5}
\end{equation}
\begin{equation}
{\left[ {\frac{{\partial {\mathbf{H}}_E^{'\; {\text{H}}}}}{{\partial \tilde r_R^l}}} \right]_{p,q}} = {\left[\mathbf{H}_E^{l} \mathbf{P}_E^{\rm H} \right]^{\rm H}_{p,q}} \times j\frac{{\pi }}{\lambda }{\text{ }} {  \frac{{{\tilde{p}^2}{d^2}{{\sin }^2}\tilde \theta _R^l}}{{{{(\tilde r_R^l)}^2}}}} ,
\label{Jac6}
\end{equation}
\begin{equation}
\frac{{\partial {\mathbf{H}}_E^{'\; {\rm H}}}}{{\partial \Re (\tilde \alpha _E^l)}} = \mathbf{P}_E {\mathbf{H}}_E^{l\;{\kern 1pt} {\text{H}}}/\tilde \alpha _E^{l\;{\kern 1pt} *},
\label{Jac7}
\end{equation}
\begin{equation}
\frac{{\partial {\mathbf{H}}_E^H}}{{\partial \Im (\tilde \alpha _E^l)}} =  - j\mathbf{P}_E{\mathbf{H}}_E^{l\;{\kern 1pt} {\text{H}}}/\tilde \alpha _E^{l\;{\kern 1pt} *},
\label{Jac8}
\end{equation}
where  $\tilde{p} = -\left(N_e - 1\right)/2 + p -1$. Finally, Eve's CRB matrix is given by ${\rm CRB}\left({\bm \zeta}_E\right) = {\rm FIM}^{-1}\left({\bm \zeta}_E\right)$, and the RCRB can be obtained accordingly.
In the subsequent simulations, we will demonstrate that Eve achieves a lower RCRB for scatterers than targets under the proposed LD scheme. However CRB only indicates the estimation performance for targets and scatterers separately and cannot capture the confounding influence caused by the scatterers in our Eve deception approach.

\section{Simulation Results}
In this section, simulation results are presented to validate the effectiveness of the proposed LD scheme. Some simulation parameters are listed in Table \ref{SimSetup}, while the location information of the transmitter, receivers, reflectors is plotted in Fig. \ref{locinfo}, {\color{black}and the markers in this figure will be applied throughout this section to represent the receivers and reflectors}. The tradeoff weight is defined as $[w_c,\,w_{ss},\,w_s]$. Other unspecified parameters will be noted in the figure captions.
\renewcommand{\arraystretch}{1.1}
\begin{table}[htbp]
	\caption{Simulation parameters setup.}
	\centering
	\begin{threeparttable}
		\begin{tabular}{*{3}{>{\centering\arraybackslash}m{2cm}>{\centering\arraybackslash}m{1.5cm}>{\centering\arraybackslash}m{4cm}}} 
			\Xhline{2pt}
			\textbf{Parameter} & \textbf{Value} & \textbf{Description} \\
			\Xhline{1pt}
			$f_c$ & 28 GHz & Carrier frequency \\
			$N_t$ &  129  & Transmit antenna number  \\
			$N_{rf}$ & 20 & RF chain number \\
			$K_c$ & 2 & Communication user number\\
			$K_r$ & 1 & Radar signal dimension \\
			$L_T$ & 3 & Target number \\
			$L_C$ & 2 & Scatter number\\
			$N_r$ &  65 &   Receive antenna number of Alice\\ 
			$N_e$ &  65 &  Receive antenna number of Eve \\ 
			$d = \lambda/2$ & 0.0054 m &  Element spacing\\
			$\kappa$ & 0.8 & Power allocation between communication and sensing\\
			$S$ & 100 & Snapshot number\\
			\Xhline{2pt}
		\end{tabular}%
	\end{threeparttable}
	\label{SimSetup}%
\end{table}%

\begin{figure}[t]
	\centering
	\includegraphics[scale=1.0]{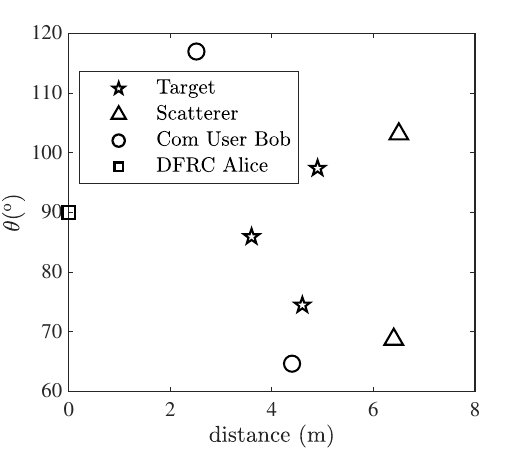}
	\caption{System location topology diagram.}
	\label{locinfo}
\end{figure}

\begin{figure}[!ht]
	\centering
	\includegraphics[scale=1.0]{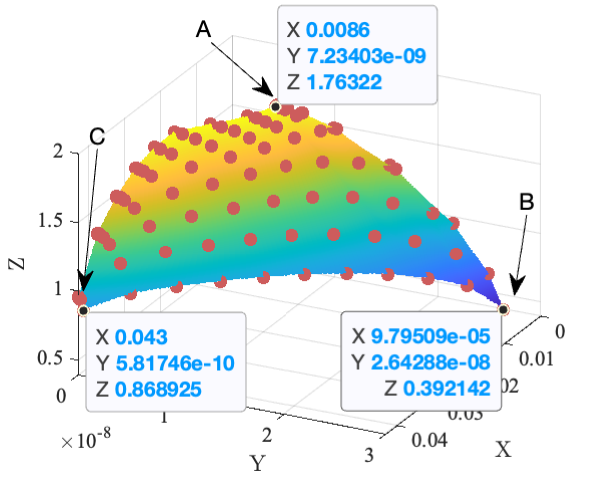}
	\caption{Three-way performance tradeoff vs. different optimization weights when SNR = 20dB (X-axes: Alice radar SINR, Y-axes: minimum power to scatterers, Z-axes: sum rate (bits/Hz/s)).}
	\label{Tradeoff}
\end{figure}

First, the tradeoff surface between the sum rate, Alice's radar SINR, and the minimal power allocated to scatterers is portrayed in Fig. \ref{Tradeoff}. Different from the conventional ISAC systems only involving a tradeoff between communication and sensing, our system exhibits a more complex three-way tradeoff relationship among communication, sensing, and security. Overall, these objectives are inherently contradictory. By analyzing the three endpoints (A, B, and C) on the tradeoff surface, we can draw several interesting conclusions. At point A, where optimal communication performance is obtained, Alice's radar SINR and the power allocated to the scatterers are not entirely zero, indicating that spatial reflectors are useful for enhancing the sum rate. At point B,  corresponding to maximal sensing security, the rate is about 0.39 while Alice's radar SINR is zero, demonstrating a strict conflict between sensing security and Alice's radar sensing performance. A similar conclusion can be drawn based on point C, emphasizing the inherent conflict between sensing performance and sensing security. Therefore, in a practical scenario, one may select an appropriate point on the tradeoff surface based on the specific scenario requirements.

To more intuitively demonstrate the beamforming effects of the proposed scheme at different compromise points, we present the transmit beampatterns at different tradeoff points in Fig. \ref{TxBP}. Fig. \ref{TxBP} (a) depicts a sensing-dominant scenario, where most power is allocated to the targets, with the farther two targets being allocated more energy due to their larger round-trip path losses. Fig. \ref{TxBP} (b) shows the result under a communication-dominant scenario, where most of the energy is allocated to the LoS path of Bob. However, reflectors also receive a portion of the power due to their contribution in enhancing spatial diversity gain \cite{FofWirel}. A Sensing-security dominant scenario is shown in Fig. \ref{TxBP} (c), where most power is allocated to scatterers.
Finally, Fig. \ref{TxBP} (d) illustrates a compromise scenario that considers all requirements, demonstrating that the proposed scheme is capable of striking a flexible tradeoff performance.

\begin{figure*}[!htb]
	\centering
	\begin{minipage}{0.24\textwidth}
		\centering
	\includegraphics[scale=0.8]{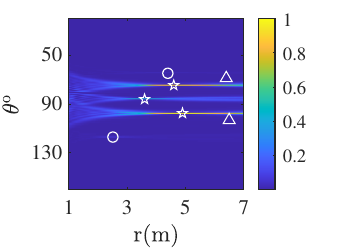}
		\subcaption*{(a)}
		\vspace{-1mm}
	\end{minipage}
	\hfill
	\begin{minipage}{0.24\textwidth}
		\centering
	\includegraphics[scale=0.8]{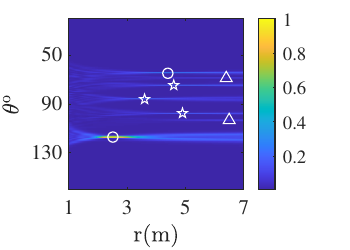}
		\subcaption*{(b)}
		\vspace{-1mm}
	\end{minipage}
	\hfill
	\begin{minipage}{0.24\textwidth}
		\centering
	\includegraphics[scale=0.8]{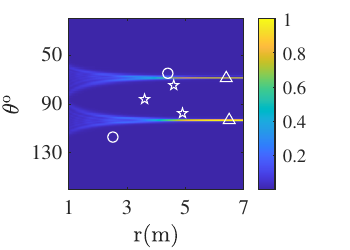}
		\subcaption*{(c)}
		\vspace{-1mm}
	\end{minipage}
		\hfill
	\begin{minipage}{0.24\textwidth}
		\centering
	\includegraphics[scale=0.8]{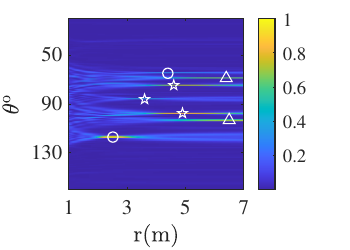}
		\subcaption*{(d)}
		\vspace{-1mm}
	\end{minipage}
	\caption{Transmit beampatterns when SNR = 20 dB: (a) sensing dominant; (b) communication dominant; (c) sensing-security dominant; (d) all considered. (circle: communication users, star: targets, triangle: scatterers)}\label{TxBP}
	\vspace{-3mm}
\end{figure*}

\setcounter{figure}{5}
\begin{figure*}[!t]
	\centering
	\begin{minipage}{0.24\textwidth}
		\includegraphics[scale=0.8]{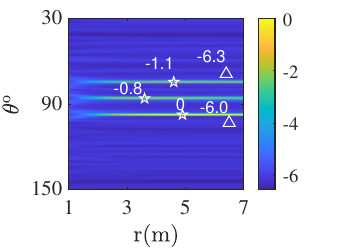}
		\subcaption*{(a)}
		\vspace{-1mm}
	\end{minipage}
	\hfill
	\begin{minipage}{0.24\textwidth}
		\includegraphics[scale=0.8]{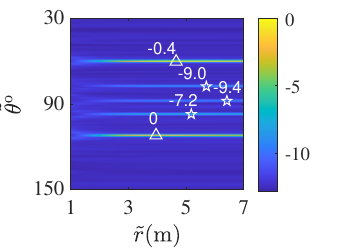}
		\subcaption*{(b)}
		\vspace{-1mm}
	\end{minipage}
	\begin{minipage}{0.24\textwidth}
		\includegraphics[scale=0.8]{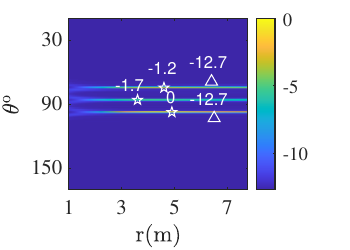}
		\subcaption*{(c)}
		\vspace{-1mm}
	\end{minipage}
	\hfill
	\begin{minipage}{0.24\textwidth}
		\includegraphics[scale=0.8]{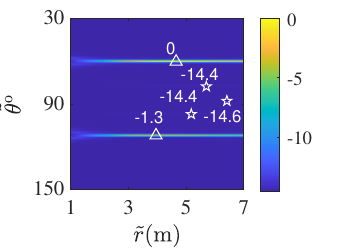}
		\subcaption*{(d)}
		\vspace{-1mm}
	\end{minipage}
	\caption{Spatial (pseudo) spectrum in dB form when SNR = 30 dB, $[w_c,\,w_{ss},\,w_s] = [0.3,\,0.7,\,0.1]$, and Eve's location $\left(90^{\rm o},\,10 {\rm m}\right)$: (a) Alice Capon; (b) Eve Capon; (c) Alice MUSIC; (d) Eve MUSIC.  (star: targets, triangle: scatterers)}\label{SpatialSpectrum}
	\vspace{-3mm}
\end{figure*}

\setcounter{figure}{4}
\begin{figure}[!t]
	\centering
	\includegraphics[scale=1.0]{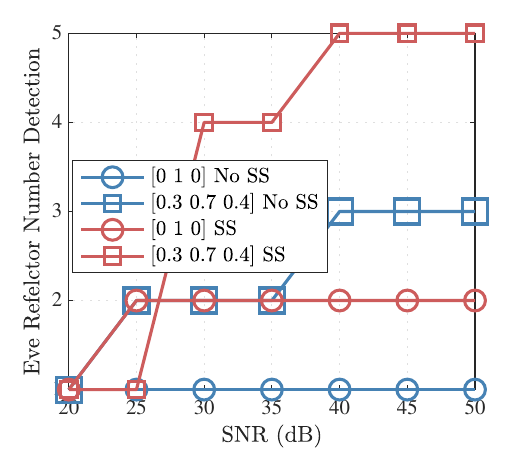}
	\caption{Estimated number of reflectors by Eve $\left(90^{\rm o},\,10 \,{\rm m}\right)$ vs. SNR under different tradeoff weight $[w_c,\,w_{ss},\,w_s]$.}\label{DetNum}
\end{figure}

\setcounter{figure}{6}
\begin{figure}[!htb]
	\centering
	\includegraphics[scale=1.0]{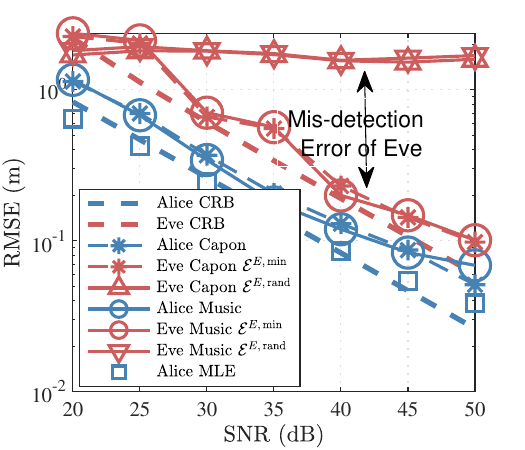}
	\caption{RMSE of range parameters vs. SNR with Eve' location $\left(90^{\rm o},\;10 {\rm m}\right)$ and $[w_c,\,w_{ss},\,w_s] = [0.3,\,0.7,\,0.4]$.}
	\label{MSECRBDis}
\end{figure}

Since the optimized sensing performance metrics are indirect, we will next validate the effectiveness of the proposed scheme using direct indicators of receiver-side sensing performance, i.e. spatial spectrum, MSE, and CRB. {\color{black}First, we analyze the Eve's performance in estimating the number of reflectors in Fig. \ref{DetNum} under different tradeoff weights. Firstly, we consider an Eve who does not employ the SS technique for decorrelation. When the tradeoff weight equals $[0,\,1\,,0]$, i.e., all power is allocated to the two scatterers and communication signal is zero, the reflector estimation remains constant at 1 as SNR varies. This result is due to the rank-1 constraint of radar covariance matrix. In contrast, when all reflectors are allocated power, i.e., $[w_c,\,w_{ss},\,w_s] = [0.3,\,0.7,\,0.4]$,  Eve's estimation of the number of reflectors is limited to a maximum of 3 in the high SNR region, rather than the actual number of 5. This discrepancy is due to the rank-3 limitation of the transmit covariance matrix. For a stronger Eve that can employ SS \cite{NFSS}, the estimation performance is improved significantly for both the two tradeoff weights. Particularly when both targets and scatterers are allocated energy, the estimation results reach 4 or even the actual 5. It is noteworthy that in this scenario, as the SNR and the snapshots number increase, Eve is ultimately able to accurately estimate the correct number of reflectors.}

In Fig. \ref{SpatialSpectrum}, we present the spatial (pseudo) spectra of Alice and Eve using the SS-assisted Capon and SS-assisted MUSIC methods, with an SNR of 30 dB and optimization weights of $[w_c,\;w_{ss},\;w_s]= [0.3,\;0.7,\;0.1]$. Initially, examining the spatial spectra obtained from the Capon estimator in Figs. \ref{SpatialSpectrum} (a) and (b), it is evident that Alice can effectively eliminate interference from the known scatterers \cite{InterCancel}, resulting in distinct peaks only at the targets' locations. In contrast, Eve exhibits very sharp peaks at the scatterers' locations, while the peaks at the locations of the target are not as prominent, {\color{black}which means that Eve achieves better sensing performance at scatterers than at targets and is more easily mistaken scatterers as real targets.} Unlike the energy spectrum of Capon, Figs. \ref{SpatialSpectrum} (c) and (d) depict the MUSIC pseudo-spectrum, which characterizes the orthogonality between the signal subspace and the noise subspace. For Alice, the overall trend is similar to that of Capon. However, for Eve, due to the incorrect estimation of the number of reflectors, the three smaller eigenvectors corresponding to the targets are mistakenly classified into the noise subspace, resulting in the complete absence of all targets in the MUSIC pseudo-spectrum. {\color{black}This means that targets are completely undetectable by Eve, thereby achieving absolute sensing security.}

\begin{figure}[!t]
	\centering
	\includegraphics[scale=1.0]{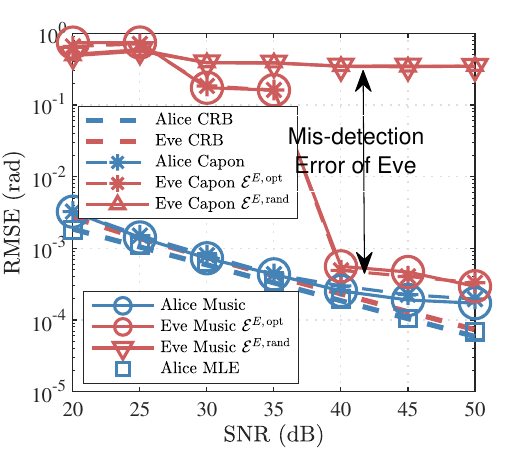}
	\caption{RMSE of angle parameters versus SNR with Eve' location $\left(90^{\rm o},\;10 {\rm m}\right)$ and $[w_c,\,w_{ss},\,w_s] = [0.3,\,0.7,\,0.4]$.}
	\label{MSECRBAng}
\end{figure}

\begin{figure}[!t]
	\centering
	\includegraphics[scale=1.0]{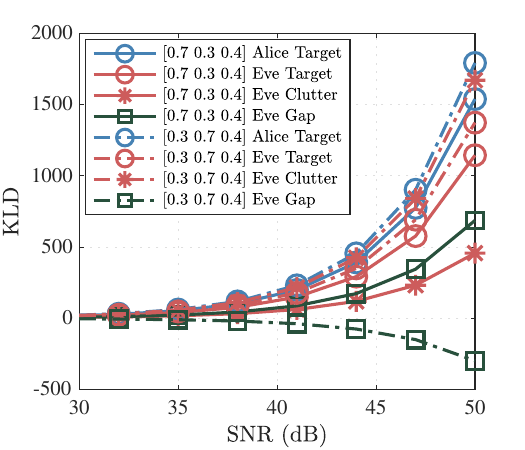}
	\caption{KLD at both Alice and Eve $\left(90^{\rm o},\,10 \,{\rm m}\right)$ versus SNR and tradeoff weight.}
	\label{KLD}
\end{figure}

\begin{figure}[!t]
	\centering
	\includegraphics[scale=1.0]{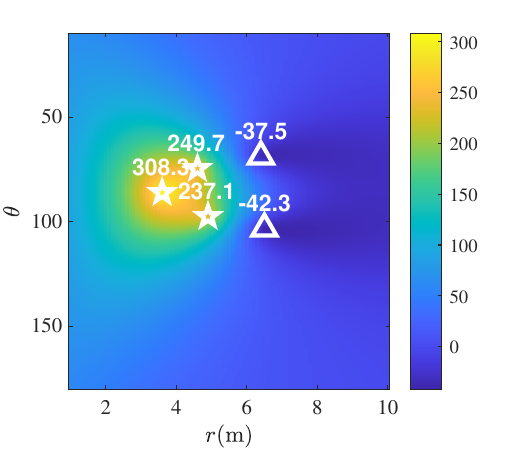}
	\caption{KLD gap of Eve versus its locations with SNR = 40 dB and $[w_c,\,w_{ss},\,w_s] = [0.3,\,0.7,\,0.4]$. (star: targets, triangle: scatterers)}
	\label{KLDvsEveLoC}
\end{figure}

Furthermore, the RCRB and RMSE versus SNR for the targets' distance and angle parameters for both Alice and Eve are plotted in Figs. \ref{MSECRBDis} and \ref{MSECRBAng} respectively with optimization weight $[w_c,\,w_{ss},\,w_s] = [0.3,\,0.7,\,0.4]$. Firstly, it is observed that the RCRB and RMSE for angle parameters are better than those for distance, owing to the lower  distance resolution in near-field especially when the antenna number of transceiver are not efficient \cite{NFChanEst}. Subsequently, by comparing the RMSE across different estimators, we observe that the MLE demonstrates the best performance, followed by MUSIC and Capon. Comparing the estimation performance between Alice and Eve, we find that their RCRB levels are close. However, when Eve employs random index $\mathcal{E}_E^{\rm rand}$, i.e., Eve is unable to distinguish between scatterers and targets, its RMSE decreases upon increasing SNR initially but then maintains at a higher error floor. In contrast, Alice's RMSE continuously decreases as SNR increases and nearly approaches the RCRB. When Eve uses ground-truth-based index  $\mathcal{E}_E^{\rm min}$, which allows it to distinguish between scatterers and targets, its RMSE initially decreases, then stabilizes, and finally reduces further to reach the RCRB as SNR increases. This trend occurs because Eve's estimated number of reflectors gradually improves until it can accurately estimate all reflectors. Therefore, when SNR is low, part of the targets remain undetected, providing absolute sensing security. However, in higher SNR regions where Eve can detect all reflector positions, sensing security stems from the confusion effect created by the scatterers. 


Finally, the KLD performance is given in Fig. \ref{KLD} to illustrate the effect of the proposed LD scheme at Eve. Specifically, we select two optimization weights  [0.7, 0.3, 0.4] and  [0.3, 0.7, 0.4], where the Alice radar SINR weight is fixed. 
Increasing the weight of the power toward scatterers from 0.3 to 0.7 results in an increase in the KLD at the scatterers, which surpasses the KLD at the targets and leads to a negative growth in the KLD gap. This indicates that the scatterers are more easily detectable compared to the targets under the influence of the proposed LD scheme. This will lead Eve to mis-detect scatterers as targets of interest.

Additionally, we observe that the impact of the proposed LD scheme on Eve varies with its location. Therefore, we present a heatmap of the Eve KLD gap versus its locations in Fig. \ref{KLDvsEveLoC}, where one can clearly see that the LD scheme performs poorly when Eve is close to the targets, while its sensing-security performance improves significantly when Eve is away from the targets and close to the scatterers. This conclusion provides valuable guidance for our future work, such as scatterers selection or RIS/relay deployment to further disrupt Eve.

\section{Competing Interests}
LD declares a relevant patent application: United Kingdom Patent Application No. 2511028.9.

\section{Conclusion}
In this paper, we introduce an LD scheme for NF-ISAC systems that utilizes the known scatterers to enhance the sensing security of the targets' location information while guaranteeing communication service. By designing and solving an optimization problem that increases the power towards the scatterers and constrains the rank of the transmit covariance matrix, we aim to deceive potential Eves in mis-detecting scatterers as targets of interest. This significantly suppresses Eve's sensing performance. Simulation results demonstrate that the proposed scheme achieves a flexible three-way tradeoff among communication, sensing, and sensing-security. Furthermore, the proposed LD scheme can lead to false or even missed detection of actual targets by Eve.

\vspace{-2mm}
\appendices
\section{Proof of Lemma \ref{lemma1}}\label{AppdixA}
Let's first assume a feasible solution $\mathbf{R}_X = \sum\limits_{k = 1}^{K_c} {{\mathbf{R}}_c^{k}}  + {\mathbf{R}}_r$, where the radar signal covariance is not zero. Then, we can always find another feasible solution ${\mathbf{\tilde R}}_X  = \sum\limits_{k = 1}^{K_c} {{\mathbf{\tilde R}}_c^{k\; }} $, where ${\mathbf{\tilde R}}_c^{k } = {\mathbf{R}}_c^{k } + \frac{1}{K}{\mathbf{R}}_r $. It is straightforward to verify that the overall transmit covariance matrix remains unchanged, i.e., ${\mathbf{\tilde R}}_X = {\mathbf{R}}_X$. Therefore, all constraints in problem $\mathcal{P}_2$ remain satisfied. On the other hand, each term in the objective function becomes
\begin{equation}
\begin{gathered}
	\frac{{{\text{Tr}}\left( {{\mathbf{\tilde R}}_c^{k }{\mathbf{H}}_B^k} \right)}}{{\sum\limits_{j \ne k} {{\text{Tr}}\left( {{\mathbf{\tilde R}}_c^{j }{\mathbf{H}}_B^k} \right) + \sigma _{\text{n}}^2} }} \hfill \\
	= \frac{{{\text{Tr}}\left( {{\mathbf{R}}_c^{k }{\mathbf{H}}_B^k} \right) + \frac{1}{K}{\text{Tr}}\left( {{\mathbf{R}}_r {\mathbf{H}}_B^k} \right)}}{{\sum\limits_{j \ne k} {{\text{Tr}}\left( {{\mathbf{R}}_c^{j }{\mathbf{H}}_B^k} \right) + \frac{{K - 1}}{K}{\text{Tr}}\left( {{\mathbf{R}}_r {\mathbf{H}}_B^k} \right) + \sigma _{\text{n}}^2} }} \hfill \\
	> \frac{{{\text{Tr}}\left( {{\mathbf{R}}_c^{k }{\mathbf{H}}_B^k} \right)}}{{\sum\limits_{j \ne k} {{\text{Tr}}\left( {{\mathbf{R}}_c^{j }{\mathbf{H}}_B^k} \right) + {\text{Tr}}\left( {{\mathbf{R}}_r {\mathbf{H}}_B^k} \right) + \sigma _{\text{n}}^2} }}, \hfill \\ 
\end{gathered} 
\end{equation}
where $\mathbf{H}_B^k = \mathbf{h}_B^k \mathbf{h}_B^{k\; {\rm H}}$. Thus, the proof is completed.

\begin{figure*}[!htbp]
	\small
	\begin{equation}
		\begin{gathered}
			L\left( {{\mathbf{R}}_c^k,\delta ,{\mu _{{l_c}}},{\lambda _{{l_t}}},\kappa ,{\mathbf{\Gamma }}_c^k} \right) = w\sum\limits_{k = 1}^K {{{\log }_2}\left( {1 + \frac{{{\text{Tr}}\left( {{\mathbf{R}}_c^k{\mathbf{H}}_B^k} \right)}}{{\sum\limits_{j \ne k} {{\text{Tr}}\left( {{\mathbf{R}}_c^j{\mathbf{H}}_B^k} \right) + \sigma _{\text{n}}^2} }}} \right)}  + \sum\limits_{k = 1}^K {{\text{Tr}}\left( {{\mathbf{R}}_c^k{\mathbf{\Gamma }}_c^k} \right)}  + \kappa \left[ {{P_S} - {\text{Tr}}\left( {{{\mathbf{R}}_X}} \right)} \right] \hfill \\ \vspace{-2mm}
			\left( {1 - w} \right)\delta  + \sum\limits_{{l_c} \in \mathcal{C}} {{\mu _{{l_c}}}\left[ {{\text{Tr}}\left( {{{\mathbf{R}}_X}{\mathbf{A}}_{{N_t}}^{{l_c}}} \right) - \delta } \right]}  + \sum\limits_{{l_t} \in \mathcal{T}} {{\lambda _{{l_t}}}\left[ {\left( {1 + {P_{{\text{SCNR}}}}} \right){\text{Tr}}\left( {{{\mathbf{R}}_X}{\mathbf{B}}_{{N_t}}^{{l_t}}} \right) - {P_{{\text{SCNR}}}}\sum\limits_{l \in \mathcal{T}} {{\text{Tr}}\left( {{{\mathbf{R}}_X}{\mathbf{B}}_{{N_t}}^l} \right)} } \right]},  \hfill \\ 
		\end{gathered} 
		\label{LagFun}
	\end{equation}
	\hrulefill
	\vspace{-3mm}
\end{figure*}

Next, we prove that rank-1 communication covariance matrices exist when the radar signal covariance matrix is zero. In this case, the Lagrangian function is provided in (\ref{LagFun}), where $\mu_{l_c}$, $\lambda_{{l_t}}$, $\varsigma$, and $\mathbf{\Gamma}_c^k$ are Lagrange multipliers, and ${\mathbf{A}}_{{N_t}}^{{l_c}} $ and ${\mathbf{B}}_{{N_t}}^l $ are expressed as

\begin{equation}
{\mathbf{A}}_{{N_t}}^{{l_c}} = {\left( {\sqrt {{\rho _0}} /r_R^{{l_c}}} \right)^2}{\mathbf{a}}_{{N_t}}^{}\left( {\theta _R^{{l_c}},r_R^{{l_c}}} \right){\mathbf{a}}_{{N_t}}^{\text{H}}\left( {\theta _R^{{l_c}},r_R^{{l_c}}} \right),
\end{equation}
\begin{equation}
{\mathbf{B}}_{{N_t}}^l = {\left| {\alpha _A^l} \right|^2}{\mathbf{a}}_{{N_t}}^{}\left( {\theta _R^l,r_R^l} \right){\mathbf{a}}_{{N_t}}^{\text{H}}\left( {\theta _R^l,r_R^l} \right).
\end{equation}
Then, part of the Karush-Kuhn-Tucker (KKT) conditions are listed as
\begin{equation}
\left\{ \begin{gathered}
	{\mathbf{\Gamma }}_c^{k\;{\kern 1pt}  \star } = \varsigma {\mathbf{I}} + {{\mathbf{T}}_1} - {{\mathbf{T}}_2}, \hfill \\
	{\mathbf{R}}_c^{k\;{\kern 1pt}  \star } \succeq {\mathbf{0}},{\mathbf{\Gamma }}_c^{k\;{\kern 1pt}  \star } \succeq {\mathbf{0}},{\mu _{{l_c}}} \geqslant 0,{\lambda _{{l_t}}} \geqslant 0,\varsigma  \geqslant 0, \hfill \\
	{\text{Tr}}\left( {{\mathbf{\Gamma }}_c^{k\;{\kern 1pt}  \star }{\mathbf{R}}_c^{k\;{\kern 1pt}  \star }} \right) = 0, \hfill \\ 
\end{gathered}  \right.
\end{equation}
where $\mathbf{T}_1$ and $\mathbf{T}_2$ are expressed as
\begin{equation}
\begin{gathered}
	{{\mathbf{T}}_1} = \frac{w}{{\ln 2}} \cdot \sum\limits_{j \ne k} {\frac{1}{{\sum\limits_{i \ne j,k} {{\text{Tr}}\left( {{\mathbf{R}}_c^i{\mathbf{H}}_B^k} \right)}  + {\text{Tr}}\left( {{\mathbf{R}}_c^{k\;{\kern 1pt}  \star }{\mathbf{H}}_B^k} \right) + \sigma _{\text{n}}^2}}}  \times  \hfill \\
	\frac{{{\text{Tr}}\left( {{\mathbf{R}}_c^j{\mathbf{H}}_B^j} \right)}}{{\sum\limits_{i \ne k} {{\text{Tr}}\left( {{\mathbf{R}}_c^i{\mathbf{H}}_B^k} \right)}  + {\text{Tr}}\left( {{\mathbf{R}}_c^{k\;{\kern 1pt}  \star }{\mathbf{H}}_B^k} \right) + \sigma _{\text{n}}^2}} \cdot {\mathbf{H}}_B^j  -  \sum\limits_{{l_c} \in \mathcal{C}} {{\mu _{{l_c}}}{\mathbf{A}}_{{N_t}}^{{l_c}}} \hfill \\
	 - \sum\limits_{{l_t} \in \mathcal{T}} {{\lambda _{{l_t}}}\left[ {\left( {1 + {P_{{\text{SCNR}}}}} \right){\mathbf{B}}_{{N_t}}^{{l_t}} - {P_{{\text{SCNR}}}}\sum\limits_{l \in \mathcal{T}} {{\mathbf{B}}_{{N_t}}^l} } \right]},  \hfill \\ 
\end{gathered} ,
\end{equation}
\begin{equation}
{{\mathbf{T}}_2} = \frac{w}{{\ln 2}} \cdot \frac{{{\mathbf{H}}_B^k}}{{{\text{Tr}}\left( {{\mathbf{R}}_c^{k\;{\kern 1pt}  \star }{\mathbf{H}}_B^k} \right) + \sum\limits_{j \ne k} {{\text{Tr}}\left( {{\mathbf{R}}_c^j{\mathbf{H}}_B^k} \right)}  + \sigma _{\text{n}}^2}}.
\end{equation}
It is noted that $\mathbf{T}_2$ is a rank-1 matrix. Thus, we have the similar KKT conditions structures as those in \cite{Rank1Proof}. Using the Proof of Proposition 4.1 in \cite{Rank1Proof}, we can complete the proof.

\bibliographystyle{IEEEtran}      
\bibliography{IEEEabrv,references}

\end{document}